%% file: paper.tex
\documentclass{llncs}[10pt]
\usepackage[bindingoffset=0in,left=1in,right=1in,top=1in,bottom=1in,footskip=.25in]{geometry}

\usepackage{makeidx}  % allows for indexgeneration

\usepackage{graphicx}
\usepackage{setspace,algorithm,algorithmic,amsfonts}
\usepackage{amssymb, amsmath}
\usepackage{color}
\usepackage{verbatim}

\spnewtheorem{defi}{Definition}{\bfseries}{\rmfamily}
\spnewtheorem{redrule}{Reduction Rule}{\bfseries}{\rmfamily}
\spnewtheorem{obs}{Observation}{\itshape}{\rmfamily}
\spnewtheorem{cl}{Claim}{\bfseries}{\rmfamily}
\spnewtheorem{coro}{Corollary}{\bfseries}{\rmfamily}
\spnewtheorem{transf}{Transformation}{\bfseries}{\rmfamily}

\spnewtheorem*{claimproof}{Proof}{\itshape}{\rmfamily}

\begin{document}

\pagestyle{headings}  %For number of the pages.

%temporary title...
\title{Kernelization Algorithms for Packing Problems Allowing Overlaps (Extended Version)}
%\title{Kernelization Algorithms for Packing Problems Allowing Overlaps}

\author{Henning Fernau\inst{1}, Alejandro L\'opez-Ortiz\inst{2}, and Jazm\'in Romero\inst{2}}
\institute{FB 4-Abteilung Informatikwissenschaften, Universit\"{a}t Trier, Germany.
\and
David R.\ Cheriton School of Computer Science, University of Waterloo, Canada.}
\maketitle

%%% VERSION TO ARXIV %%%%%%%%%%%%%%%%%%%%%%%%%%%%%%%%%%%%%%%%%%%%%%%%%%%%%%%%%%%

\begin{abstract}
\input{Abstract}

\end{abstract}

%***********************************************************************************
%<----------------------------------------INTRODUCTION -------------------------------
\section{Introduction}\label{introductionSection}
\input{Introduction}

%************************************************************************************
%<----------------------------------------NOTATION ------------------------------------
\section{Terminology}\label{terminologySection}
\input{Terminology}

%************************************************************************************
%<----------------------------------------COMPLEXITY ------------------------------------
\section{Hardness of Packing Problems Allowing Overlaps}\label{complexityResults}
\input{Complexity}

%************************************************************************************
%<----------------------------------------RESULTS----------------------------------

%< Membership Problems-------------------------------------------------------------
%*************************************************************************************
\section{Packing Problems with Bounded Membership}\label{membershipSection}

In this section, we introduce a polynomial parametric transformation for the $r$-Set Packing with $t$-Membership problem to the {$(r+1)$}-Set Packing problem. We obtain a kernel result for the $r$-Set Packing with $t$-Membership problem by running a kernelization algorithm on the transformed instance of the Set Packing problem. This compression result can be turned into a proper kernel result by re-interpreting the $(r+1)$-Set Packing kernel within the original problem. 

To obtain kernel results for our graph with $t$-Membership problems, we transform each graph with $t$-Membership problem to an instance of the $r$-Set Packing with $t$-Membership problem. Recall that $r,t \geq 1$ are fixed constants.

%Packing Sets------------------------------------------------------------>
\subsection{Packing Sets with $t$-Membership}\label{membershipSets}

Let us begin with the $r$-Set Packing with $t$-Membership problem. 
%Let $r,t\geq 1$ be fixed in the following problem definition, that in actual fact defines a whole family of problems.

\input{Membership_SetPacking_ToSP}

%Packing Graphs---------------------------------------------------------->
\subsection{Packing Graphs with $t$-Membership - Vertex-Membership}\label{membershipGraphs}
%\subsection{Packing Graphs with $t$-Membership}\label{membershipGraphs}
We introduce the $\mathcal{H}$-Packing with $t$-Membership problem to bound the number of $\mathcal{H}$-subgraphs that a vertex of $G$ can belong to. To reduce this problem to a kernel, we will transform it to the $r$-Set Packing with $t$-Membership. 

%In the following family of definitions, again $t\geq 1$.

\input{Membership_Vertex_ToMembershipSP}

\subsection{Packing Graphs with $t$-Membership - Edge-Membership}\label{EdgeMembershipSection}
%\subsubsection{Edge Membership.}

\input{Membership_Edge_ToMembershipSP}

%*****************************************************************************
%< Pairwise Problems--------------------------------------------------------->
%*****************************************************************************

\section{Packing Problems with Bounded Overlap}\label{pairwiseSection}

%*****************************************************************************

In this section, we provide kernel results for our $t$-Overlap problems. First, we develop a kernelization algorithm for the problem of packing sets with pairwise overlap (called $r$-Set Packing with $t$-Overlap). Then, we transform each graph with $t$-Overlap problem to an instance of the $r$-Set Packing with $t$-Overlap problem. These transformations allow us to obtain kernel results for all our graph with $t$-Overlap problems, as we can always recover an instance of the graph-theoretic question from the reduced $r$-Set Packing with $t$-Overlap instance. 

Once more, $r \geq 1$ and  $t\geq 0$ are fixed constants in our problems.

%< Kernel for Set Packing----------------------------------------------------->
\subsection{Packing Sets with $t$-Overlap}\label{kernelalgorithm}

We begin by developing a kernelization algorithm for the $r$-Set Packing with $t$-Overlap problem. 
%, where $r \geq 1$ and  $t\geq 0$ are fixed in the following definition

\input{Pairwise_SetPacking_Kernel}

%*************************************************************************
%<Pairwise Overlap for Graph Packing-------------------------------------->
\subsection{Packing Graphs with $t$-Overlap}\label{pairwisepacking}

In this subsection, we provide kernel results for graph packing problems that allow pairwise overlap. To this end, we first transform 
the $\mathcal{H}$-Packing with $t$-Overlap problem (and its variants) to an instance of the $r$-Set Packing with $t$-Overlap problem. Then, we run kernelization algorithm \ref{iterative} (Section \ref{kernelalgorithm}). 
Finally, we re-interpret that kernel as a kernel for the original (graph) problem. 
%Let $t\geq 0$ in the following definition.

%<Vertex version ----------------------------------------------------------------->
\input{Pairwise_Vertex_ToPairwiseSP}

%<Induced Vertex version --------------------------------------------------------->
\input{Pairwise_Induced_ToPairwiseSP}

%<Edge Packing ------------------------------------------------------------------->
\subsubsection{Edge-Overlap}\label{edgeoverlap}
\input{Pairwise_Edge_ToPairwiseSP}

%<Family of Cliques--------------------------------------------------------------->
\subsection{The Family of Cliques}
\input{Pairwise_Cliques_ToPairwiseSP}

\section{Conclusions}
\input{Conclusions.tex}

%<-----------------REFERENCES---------------------------------------------------->
\bibliographystyle{splncs03}
\bibliography{biblio}

\end{document}

%% file: Abstract.tex
We consider  the problem of discovering overlapping communities in networks which we model as 
generalizations of the Set and Graph Packing problems with overlap.
As usual for Set Packing problems we seek a collection $\mathcal{S}' \subseteq \mathcal{S}$  
consisting of at least $k$ sets subject to certain disjointness restrictions. 
In the $r$-Set Packing with $t$-Membership, each element of $\mathcal{U}$ belongs to at most $t$ sets of $\mathcal{S'}$ while in $r$-Set Packing with $t$-Overlap each pair of sets in $\mathcal{S'}$ overlaps in at most $t$ elements. For both problems, each set of $\mathcal{S}$ has at most $r$ elements.%, and $t$ and $r$ are constants. 

Similarly, both of our graph packing problems seek a collection $\mathcal{K}$ of at least $k$ subgraphs in a graph $G$  each isomorphic to a graph $H \in \mathcal{H}$.
In $\mathcal{H}$-Packing with $t$-Membership, each vertex of $G$ belongs to at most $t$ subgraphs of $\mathcal{K}$ while in $\mathcal{H}$-Packing with $t$-Overlap each pair of subgraphs in $\mathcal{K}$ overlaps in at most $t$ vertices. For both problems, each member of $\mathcal{H}$ has at most $r$ vertices and $m$ edges, where  $t$, $r$, and $m$ are constants.

Here, we show NP-Completeness results for all of our packing problems. Furthermore, we give a dichotomy result for the $\mathcal{H}$-Packing with $t$-Membership problem analogous to the Kirkpatrick and Hell~\cite{Kirk78}. Given this intractability, we reduce the $r$-Set Packing with $t$-Membership to a problem kernel with $O((r+1)^r k^{r})$ elements while we achieve a kernel with $O(r^r k^{r-t-1})$ elements for the $r$-Set Packing with $t$-Overlap. In addition, we reduce the $\mathcal{H}$-Packing with $t$-Membership and its edge version to problem kernels with $O((r+1)^r k^{r})$ and $O((m+1)^{m} k^{{m}})$ vertices, respectively.  On the other hand, we achieve kernels with $O(r^r k^{r-t-1})$ and $O(m^{m} k^{m-t-1})$ vertices for the $\mathcal{H}$-Packing with $t$-Overlap and its edge version, respectively. In all cases, $k$ is the input parameter while $t$, $r$, and $m$ are constants.

%% file: Introduction.tex
Networks are commonly used to model complex systems that arise in real life, for example, social and protein-interaction networks.  A community emerges in a network when two or more entities have common interests, e.g., groups of people or related proteins. Naturally, a given person can have more than one social circle, and a protein can belong to  more than one protein complex. Therefore,  communities can share members \cite{Palla05}. The way a community is modeled is highly dependent on the specific application being modeled by the network. One flexible approach is the use of a family of graphs where each graph would be a community model. 

The $\mathcal{H}$-Packing with $t$-Overlap problem was introduced in \cite{Packing14} to capture the problem of discovering overlapping communities. The goal is to find at least $k$ subgraphs (the communities) in a graph $G$ (the network) where each subgraph is isomorphic to a member of a family of graphs $\mathcal{H}$ (the community models) and each pair of subgraphs overlaps in at most $t$ vertices (the shared members). Such type of overlap can be found in certain clustering algorithms as, e.g., in \cite{BanKhu2001}.

In this work we also consider the $\mathcal{H}$-Packing with $t$-Membership problem which bounds the number of communities that a member of a network can belong to. The goal is to find in $G$ at least $k$ isomorphic subgraphs to a member of $\mathcal{H}$ such that each vertex in $V(G)$ belongs to at most $t$ isomorphic subgraphs. This type of overlap was also previously studied, for instance, in \cite{Fellows09} in the context of graph editing. 
Both of our Graph Packing problems generalize the $H$-Packing problem which consists in finding within $G$ at least $k$ vertex-disjoint subgraphs isomorphic to a graph $H$. 

Similarly, we also consider overlap for the Set Packing problem. Given a collection of sets $\mathcal{S}$ each of size at most $r$ drawn from a universe $\mathcal{U}$, the $r$-Set Packing with $t$-Overlap problem seeks for at least $k$ sets from $\mathcal{S}$ such that each pair of selected sets overlaps in at most $t$ elements. In contrast, in an $r$-Set Packing with $t$-Membership, each member of $\mathcal{U}$ is contained in at most $t$ of the $k$ sets. The pair $\mathcal{S}$ and $\mathcal{U}$ can be treated as an hypergraph, where the vertices are the members of $\mathcal{U}$ and the hyper-edges are the members of $\mathcal{S}$. Thus, an $r$-Set Packing with $t$-Overlap can be seen as $k$ hyper-edges that pairwise intersect in no more than $t$ vertices while an $r$-Set Packing with $t$-Membership can be interpreted as $k$ hyper-edges where every vertex is contained in at most $t$ of them. We believe that these generalizations are interesting on its own. 

Some of our generalized problems are NP-complete. This follows immediately from the NP-completeness of the classical  $H$-Packing and Set Packing problems,
which is a special case of our problem setting. Our goal is to design \emph{kernelization algorithms}; that is, algorithms that in polynomial time reduce any instance to a size bounded by $f(k)$ (a \emph{problem kernel}), where $f$ is some arbitrary computable function. After that a brute-force search on the kernel gives a solution in  $g(k) n^{O(1)}$ running time, where $g$ is some computable function. An algorithm that runs in that time is a \emph{fixed-parameter algorithm},
giving rise to the complexity class FPT. For all our problems, we consider $k$ as the parameter, while $t$ and $r$ are constants.

\paragraph{Related Results.}
An FPT-algorithm for the $\mathcal{H}$-Packing with $t$-Overlap problem was developed in \cite{BstAlgorithm14}. 
The running time is $O(r^{rk}k^{(r-t-1)k+2}n^r)$, where $r=|V(H)|$ for $H \in \mathcal{H}$, $|\mathcal{H}|=1$, and $0 \leq t < r$. A $2(rk-r)$ kernel when $K_r \in \mathcal{H}$, $|\mathcal{H}|=1$, and $t=r-2$ is given in \cite{Packing14}. 

The smallest kernel for the $H$-Packing problem has size $O(k^{r-1})$ where $H$ is an arbitrary graph and $r=|V(H)|$ \cite{Moser09}. More kernels results when $H$ is a prescribed graph can be found in \cite{Chen12,Fellows04a,Fernau09,JanMar2014,Prieto06,Wang10}. 

There is an $O(r^{rk}k^{(r-t-1)k+2}n^r)$ algorithm for the $r$-Set Packing with $t$-Overlap~\cite{BstAlgorithm14} while the Set Packing (element-disjoint) has an $O(r^r k^{r-1})$ kernel \cite{Faisal10}, also confer to~\cite{DelMar2012}.

\paragraph{Summary of Results.}
For our Set Packing generalizations, each member in $\mathcal{S}$ has at most $r$ elements while for our Graph Packing problems, each member of $\mathcal{H}$ has at most $r$ vertices and $m$ edges. For the problems with $t$-Membership, $t \geq 1$ while for $t$-Overlap problems, $0 \leq t \leq r-1$. In any case, $t$, $r$, and $m$ are constants. The induced versions of our graph problems seek for induced isomorphic subgraphs in $G$, while the edge versions focus the overlap on edges instead of vertices.

In 
Section \ref{complexityResults}, we show that the $\mathcal{H}$-Packing with $t$-Membership problem is $NP$-Complete for all values of $t \geq 1$ when $\mathcal{H}=\{H\}$ and $H$ is an arbitrary connected graph with at least three vertices, but
polynomial-time solvable for smaller graphs.
%. On the other hand, we show that the $\{P_2\}$-Packing with $t$-Membership problem can be solved in polynomial time. This polynomial time result together with our NP-completeness result lead us to 
Hence, we obtain a dichotomy result for the $\{H\}$-Packing with $t$-Membership which is analogous to the one of Kirkpatrick and Hell~\cite{Kirk78}. In addition, we show NP-completeness results for the induced and edge versions of $\mathcal{H}$-Packing with $t$-Membership and for $r$-Set Packing with $t$-Membership. Moreover, we prove that for any $t \geq 0$, there always exists a connected graph $H_t$ such that the $\mathcal{H}$-Packing problem with $t$-Overlap is NP-Complete where $\mathcal{H}=\{H_t\}$.

In Section \ref{membershipSection}, we give a polynomial parameter transformation (PPT) from the $r$-Set Packing with $t$-Membership problem to an instance of the $(r+1)$-Set Packing problem. With this transformation, we reduce $r$-Set Packing with $t$-Membership to a problem kernel with $O((r+1)^r k^{r})$ elements. In addition, we obtain a kernel with $O((r+1)^r k^{r})$ vertices for $\mathcal{H}$-Packing with $t$-Membership and for its induced version by reducing these problems to $r$-Set Packing with $t$-Membership. 
In a similar way, we obtain a kernel with $O((m+1)^{m} k^ {m})$ vertices for the edge version. PPTs are commonly used to show kernelization lower bounds \cite{Bodlaender08}. To our knowledge, this is the first time that PPTs are used to obtain problem kernels.

Inspired by the ideas in \cite{Fellows04b,Moser09} in Section \ref{pairwiseSection}, we give a kernelization algorithm that reduces $r$-Set Packing with $t$-Overlap to a kernel with $O(r^r k^{r-t-1})$ elements. By transforming an instance of $\mathcal{H}$-Packing with $t$-Overlap to an instance of $r$-Set Packing with $t$-Overlap, we achieve a kernel with $O(r^r k^{r-t-1})$ vertices for the $\mathcal{H}$-Packing with $t$-Overlap problem and for its induced version as well. With similar ideas, we obtain a kernel with $O(m^{m} k^ {m-t-1})$ vertices for the edge version. We improve this kernel to $O(r^r k^{r-t-1})$, where $\mathcal{H}$ is the family of cliques.

%% file: Terminology.tex
%Standard definition and notation
Graphs are undirected and simple unless specified otherwise. For a graph $G$, $V(G)$ and $E(G)$ denote its sets of vertices and edges, respectively. We use the letter $n$ to denote the size of $V(G)$. The subgraph induced in $G$ by a set of vertices $L$ is represented by $G[L]$. The degree of a vertex $v$ in $G$ will be denoted as $deg_{G}(v)$. A set of vertices $V'$ or edges $E'$ is \emph{contained} in a subgraph $H$, if $V \subseteq V(H)$ or $E \subseteq E(H)$. In addition, $V(E')$ denotes the set of vertices corresponding to the end-points of the edges in $E'$.

Let $\mathcal{S}$ be a collection of sets drawn from $\mathcal{U}$. The size of $\mathcal{U}$ will be denoted by $n$. For $\mathcal{S'} \subseteq \mathcal{S}$, $val(\mathcal{S'})$ denotes the union of all members of $\mathcal{S'}$. For $P \subset \mathcal{U}$, $P$ is \emph{contained} in $S \in \mathcal{S}$, if $P \subseteq S$. Two sets $S$,$S'$$\in \mathcal{S}$ \emph{overlap in} $|S \cap S'|$ elements and they \emph{conflict} if $|S \cap S'| \geq t+1$.

%Not standard definition, but more generic.
All our graph problems deal with a family of graphs $\mathcal{H}$ where each $H \in \mathcal{H}$ is an arbitrary graph. Here and in the following, let $r(\mathcal{H})=\max\{|V(H)|: H\in\mathcal{H}\}$ denote the order of the biggest graph in $\mathcal{H}$. Observe that in order to have this as a reasonable notion, we always (implicitly) require $\mathcal{H}$ to be a finite set. We simply write $r$ instead of $r(\mathcal{H})$ if $\mathcal{H}$ is clear from the context. Similarly, $m(\mathcal{H})=\max\{|E(H)|: H\in\mathcal{H}\}$, and we write $m$ whenever it is clear from the context. A subgraph of $G$ that is isomorphic to some $H \in \mathcal{H}$ is called an \emph{$\mathcal{H}$-subgraph}. We denote as $\mathcal{H}_G$ the set of all $\mathcal{H}$-subgraphs in $G$.

%A small intro to give the problem definitions. The terms object and group are only used in this paragraph.. 
We consider packing problems that allow overlap in a solution in two different ways. In the \emph{$t$-Membership problems}, the \emph{overlap condition} bounds the number of times an \emph{object} is contained in a \emph{group} of a solution. In the \emph{$t$-Overlap problems},  that condition bounds the number of \emph{objects} that a pair of \emph{groups} of a solution shares.  Depending on the type of problem (packing sets or packing graphs) an object is an element, a vertex or an edge, and a group is a set or a subgraph.

We next introduce the formal definitions of all the $t$-Membership problems considered in this work. Let $r,t\geq 1$ be fixed in the following problem definitions, that in actual fact defines a whole family of problems.

%Note for smaller margins reduce to 11.3cm
\medskip
\begin{center}
\fbox{
\parbox{13.5cm}{
\textbf{The $r$-Set Packing with $t$-Membership problem}
    
    \noindent \emph{Input}: A collection $\mathcal{S}$ of distinct sets, each of size at most $r$, drawn from a universe $\mathcal{U}$ of size $n$, and a non-negative integer $k$.

    \noindent \emph{Parameter}: $k$

    \noindent \emph{Question}:  Does $\mathcal{S}$ contain a \emph{$(k,r,t)$-set membership}, i.e., at least $k$ sets $\mathcal{K}=\{S_1,\dots , S_k\}$ where each element of $\mathcal{U}$ is in at most $t$ sets of $\mathcal{K}$% where $t \geq 1$
    ?}}
\end{center}

\medskip
\begin{center}
\fbox{
\parbox{13.5cm}{
\textbf{The $\mathcal{H}$-Packing with $t$-Membership problem}.
    
    \noindent \emph{Input}: A graph $G$, and a non-negative integer $k$.

    \noindent \emph{Parameter}: $k$

    \noindent \emph{Question}: Does $G$ contain a $(k,r,t)$-$\mathcal{H}$-\emph{membership}, i.e.,  a set of at least $k$ subgraphs $\mathcal{K}=\{H_1, \dots ,H_k\}$ where each $H_i$ is isomorphic to some graph $H \in \mathcal{H}$, $V(H_i) \neq V(H_j)$ for $i\neq j$, and every vertex in $V(G)$ is contained in at most~$t$ %$H$-
    subgraphs of~$\mathcal{K}$%  where $t \geq 1$
    ?}}
\end{center}

\medskip
\begin{center}
\fbox{
\parbox{13.5cm}{
\textbf{The Induced-$\mathcal{H}$-Packing with $t$-Membership problem}.
    
    \noindent \emph{Input}: A graph $G$, and a non-negative integer $k$.

    \noindent \emph{Parameter}: $k$

    \noindent \emph{Question}: Does $G$ contain a $(k,r,t)$-\emph{induced}-$\mathcal{H}$-\emph{membership}, i.e.,  a set of at least $k$ induced subgraphs $\mathcal{K}=\{H_1, \dots ,H_k\}$ where each $H_i$ is isomorphic to some graph $H \in \mathcal{H}$, $V(H_i) \neq V(H_j)$ for $i\neq j$, and every vertex in $V(G)$ is contained in at most~$t$ %$H$-
    subgraphs of~$\mathcal{K}$%  where $t \geq 1$
    ?}}
\end{center}

\medskip
\begin{center}
\fbox{
\parbox{13.5cm}{
\textbf{The $\mathcal{H}$-Packing with $t$-Edge Membership problem}
    
    \noindent \emph{Input}: A graph $G$, and a non-negative integer $k$.

    \noindent \emph{Parameter}: $k$

    \noindent \emph{Question}: Does $G$ contain a $(k,m,t)$-edge-$\mathcal{H}$-membership, i.e.,  a set of at least $k$ subgraphs $\mathcal{K}=\{H_1, \dots ,H_k\}$ where each $H_i$ is isomorphic to some graph $H \in \mathcal{H}$, $E(H_i) \neq E(H_j)$ for $i\neq j$, and every edge in $E(G)$ belongs to at most~$t$ %$H$-
    subgraphs of~$\mathcal{K}$%, where $t \geq 1$
    ?
}}
\end{center}

%\textbf{HF} I did some edits. It should be ok now.

%\textcolor{blue}{Added the edge condition. If we want to have both conditions in the same problem then Algorithm that uses Transformation \ref{edgePackingToOSP} (transformation to Set Packing with Membership) will not output a solution as there may be in the solution of that algorithm, two $\mathcal{H}$-subgraphs with identical set of vertices. This is achieved in the ``Alternative..section''}

\medskip

%Notice that when $t=1$, we are back to the classical $r$-Set Packing, $\mathcal{H}$-Packing and Edge-$\mathcal{H}$-Packing problems, respectively.

Our $t$-Overlap problems are defined next. In these definitions, $r \geq 1$ and  $t\geq 0$ are fixed.

\medskip

\begin{center}
\fbox{
\parbox{13.5cm}{
\textbf{The $r$-Set Packing with $t$-Overlap problem}
    
    \noindent \emph{Instance}: A collection $\mathcal{S}$ of distinct sets, each of size at most $r$, drawn from a universe $\mathcal{U}$ of size $n$, and a non-negative integer $k$.

    \noindent \emph{Parameter}: $k$.

    \noindent \emph{Question}: Does $\mathcal{S}$ contain a $(k,r,t)$-set packing, i.e., a collection of at least $k$ sets $\mathcal{K}=\{S_1,\dots , S_k\}$ where $|S_i \cap S_j| \leq t$, for any pair $S_i,S_j$ with $i\neq j$?
}
}
\end{center}

\medskip
\begin{center}
\fbox{
\parbox{13.5cm}{
\textbf{The $\mathcal{H}$-Packing with $t$-Overlap problem}
    
    \noindent \emph{Input}: A graph $G$, and a non-negative integer $k$.

    \noindent \emph{Parameter}: $k$

    \noindent \emph{Question}: Does $G$ contain a $(k,r,t)$-$\mathcal{H}$-packing, i.e., a set of at least $k$ subgraphs $\mathcal{K}=\{H_1, \dots ,H_k\}$ where each $H_i$ is isomorphic to some graph from  $\mathcal{H}$ and $|V(H_i) \cap V(H_j)| \leq t$ for any pair $H_i,H_j$ with $i\neq j$?%where $t \geq 0$?
}}
\end{center}

\medskip
\begin{center}
\fbox{
\parbox{13.5cm}{
\textbf{The Induced-$\mathcal{H}$-Packing with $t$-Overlap problem}
    
    \noindent \emph{Input}: A graph $G$, and a non-negative integer $k$.

    \noindent \emph{Parameter}: $k$

    \noindent \emph{Question}: Does $G$ contain a $(k,r,t)$-\emph{induced}-$\mathcal{H}$-packing, i.e., a set of at least $k$ subgraphs $\mathcal{K}=\{H_1, \dots ,H_k\}$ where each $H_i$ is isomorphic to some graph from  $\mathcal{H}$ and $|V(H_i) \cap V(H_j)| \leq t$ for any pair $H_i,H_j$ with $i\neq j$?%where $t \geq 0$?
}}
\end{center}

\medskip
\begin{center}
\fbox{
\parbox{13.5cm}{
\textbf{The $\mathcal{H}$-Packing with $t$-Edge-Overlap problem}
    
    \noindent \emph{Input}: A graph $G$, and a non-negative integer $k$.

    \noindent \emph{Parameter}: $k$

    \noindent \emph{Question}: Does $G$ contain a $(k,m,t)$-edge-$\mathcal{H}$-packing, i.e., a set of at least $k$ subgraphs $\mathcal{K}=\{H_1, \dots ,H_k\}$ where each $H_i$ is isomorphic to a graph $H \in \mathcal{H}$, and $|E(H_i) \cap E(H_j)| \leq t$ for any pair $H_i,H_j$ with $i\neq j$?% where $t \geq 0$?
}}
\end{center}
\medskip

%Observe that for $t=0$, we are back to the classical $r$-Set Packing, $\mathcal{H}$-Packing and Edge-$\mathcal{H}$-Packing problems, respectively.

Notice that when $t=1$ and $t=0$ for the $t$-Membership and $t$-Overlap problems, respectively, we are back to the classical $r$-Set Packing, $\mathcal{H}$-Packing and Edge-$\mathcal{H}$-Packing problems.

Sometimes the size $l$ will be dropped from the notation $(l,r,t)$-set membership, $(l,r,t)$-$\mathcal{H}$-membership, $(l,r,t)$-set packing and $(l,r,t)$-$\mathcal{H}$-packing and be simply denoted as  $(r,t)$-set membership, $(r,t)$-$\mathcal{H}$-membership, $(r,t)$-set packing and $(r,t)$-$\mathcal{H}$-packing, respectively.

An $(r,t)$-$\mathcal{H}$-membership $P$ of $G$ is \emph{maximal} if every $\mathcal{H}$-subgraph of $G$ that is not in $P$ has at least one vertex $v$ contained in $t$ $\mathcal{H}$-subgraphs of $P$. Similarly, an $(r,t)$-set packing $\mathcal{M}$ of $\mathcal{S'} \subseteq \mathcal{S}$ is \emph{maximal} if any set of $\mathcal{S'}$ that is not already in $\mathcal{M}$ conflicts with some set in $\mathcal{M}$. That is, for each set $S \in \mathcal{S'}$ where $S \notin \mathcal{M}$, $|S \cap S'| \geq t+1$ for some $S' \in \mathcal{M}$.

%% file: Complexity.tex
Strictly speaking, we have introduced whole families of problems. Each $\mathcal{H}$ and each $t$ specify a particular 
 $\mathcal{H}$-Packing with $t$-Overlap respectively $t$-Membership problem. Notice that the constant $r$ we referred to above is rather a property of $\mathcal{H}$. 
Generally speaking, all these problems are in NP, and some of them are NP-complete. This follows since every instance of the $\mathcal{H}$-Packing problem, which is NP-complete \cite{Kirk78}, is mapped to an instance of these problems by setting $t=0$ and $t=1$, respectively. Similar reductions to the Set Packing problem apply for the $r$-Set Packing with $t$-Overlap and $t$-Membership problems. Notice that not every $\mathcal{H}$-Packing problem is NP-hard, a famous exception being 
 the $\{K_2\}$-Packing problem, also known as Maximum Matching.
 Notice that for specific sets $\mathcal{H}$ (and hence for specific $r$) and also for specific $t$, it is wide open whether the corresponding packing problems with $t$-Membership or with $t$-Overlap are NP-hard. 

Let us present first one concrete $\mathcal{H}$-Packing with $t$-Membership problem that is hard for all possible values of $t$.

\begin{theorem}\label{thm-P_3-membership}
For all $t\geq 1$, the $\{P_3\}$-Packing with $t$-Membership problem is NP-complete.
\end{theorem}

\begin{proof}
$\{P_3\}$-Packing with $t$-Membership in NP is easy to see for any fixed $t$. Moreover, the assertion is well-known for $t=1$.

We are now showing a reduction from  $\{P_3\}$-Packing with $t$-Membership to $\{P_3\}$-Packing with $(t+1)$-Membership, which proves the claim by induction.

Let $(G,k)$ be an instance of the $\{P_3\}$-Packing with $t$-Membership problem, where $G=(V,E)$, with $V=\{v_0,\dots,v_{n-1}\}$.
Without loss of generality, assume that $n$ is divisible by $t+1$; otherwise, simply add some isolated vertices to obtain an equivalent instance.
Let $U=\{u_0,\dots,u_{(2n)/(t+1)-1}\}$ be a set of new vertices, i.e., $V\cap U=\emptyset$.
Create a graph $G'=(V',E')$ as follows:
\begin{itemize}
\item $V'=V\cup U$;
\item $E'=E\cup \{v_{i+j}u_{2i/(t+1)},u_{2i/(t+1)}u_{2i/(t+1)+1}\mid 0\leq i< n, {{i\pmod {t+1}}\equiv 0},0\leq j\leq t\}$.
\end{itemize}

(See Figure \ref{construction})

We now prove the following claim: $G$ has a $(3,t)$-$\{P_3\}$-membership of size at least $k$ if and only if $G'$ has a $(3,t+1)$-$\{P_3\}$-membership of size at least $n+k$, where $n=|V|$.

If $P$ is a $(3,t)$-$\{P_3\}$-membership set for $G$, then $P'=P\cup X$ is a $(3,t+1)$-$\{P_3\}$-membership for $G'$,
where 
$$X=\{v_{i+j}u_{2i/(t+1)}u_{{2i/(t+1)}+1}\mid 0\leq i< n, {{i\pmod{t+1}}\equiv 0},0\leq j\leq t\}$$
Clearly, $|P'|=|P|+|X|=|P|+n$.

Conversely, let $P'$ be a $(3,t+1)$-$\{P_3\}$-membership for $G'=(V',E')$.
Let us denote by $oc_{P'}:V'\to\{0,\dots,t+1\}$ the \emph{occupancy} of $v$, meaning that
$$oc_{P'}(v)=\{p\in P'\mid v\text{ occurs on path }p\}.$$
Clearly, $oc_{P'}(u_{2i/(t+1)})\geq 
oc_{P'}(u_{{2i/(t+1)}+1})$ for all $0\leq i< n$, ${{i\pmod {t+1}}\equiv 0}$, $0\leq j\leq t$.
If, for some $i$, $oc_{P'}(u_{2i/(t+1)})>
oc_{P'}(u_{{2i/(t+1)}+1})$, then there must be some $p\in P'$ that looks like
$p=v_{i+j}u_{2i/(t+1)}v_{i+\ell}$ or $p=v_{i+j}v_{i+\ell}u_{2i/(t+1)}$ or
$p=v_{i+\ell}v_{i+j}u_{2i/(t+1)}$
with $j<\ell$, 
or like $p=v_{i'+j'}v_{i+\ell}u_{2i/(t+1)}$
for some $i',j'$ such that $i'\neq i$.

%\textcolor{blue}{Question 1. I believe the case $p=v_{i+s}v_{i+j}u_{2i/(t+1)}$ with $s<j$ could also exists (assuming an edge between $(v_{i+s},v_{i+j})$. Or not?}

%\textcolor{blue}{Question: Please see Figure 2. The path C $v_{i+\ell}v_{i+j}u_{2i/(t+1)}$ cannot be replaced with $v_{i+\ell}u_{2i/(t+1)}u_{{2i/(t+1)}+1}$}.
%HF: Thanks, we omitted one case that I added; notice that the case that I added is actually what you drew, not the case that you wrote in blue.

In either case, replace $p$ by 
$v_{i+\ell}u_{2i/(t+1)}u_{{2i/(t+1)}+1}$.
Doing this consecutively wherever possible, we end up with a $P'$ whose occupancy function satisfies
$oc_{P'}(u_{2i/(t+1)})=
oc_{P'}(u_{{2i/(t+1)}+1})$ for all $i$;
obviously this new $(3,t+1)$-$\{P_3\}$-membership is not smaller than the original one.
%Moreover, 

Now, $P'$ contains only two types of paths: those (collected in $P_G$) completely consisting of vertices from $G$, and those (collected in $P_X$) containing  exactly one vertex from $G$ (and two from $U$). We will maintain this invariant in the following modifications.
% We greedily select a $t$-Membership for $G$ from $P_G'$, yielding a set $P_G$.
%Let $P_G''$ collect those paths from
%$P_G'$ not contained in $P_G$. 
Also, we will maintain that $|P'|\geq n+k$ and that 
$oc_{P'}(u_{2i/(t+1)})=
oc_{P'}(u_{{2i/(t+1)}+1})$ for all $i$.
We are now describing a loop in which we gradually modify 
%the choices of 
$P'$. % and $P_G$.

(*) It might be that $P'$ is not maximal  at this stage, meaning that 
more $P_3$'s can be added to $P'$ without destroying the property that
$P'$ forms a $(3,t+1)$-$\{P_3\}$-membership. So, we first greedily add paths to $P'$ to ensure maximality. %Accordingly, we might be able to add some more paths to $P_G$. 
For simplicity, we keep the names $P_X$ and $P_G$ %, $P_G''$
for the possibly new sets.

Consider some vertex $v$ with  $oc_{P_G}(v)<t+1$. As $P'$ is maximal, $oc_{P_X}(v)=1$. 
But, if $oc_{P_G}(v)=t+1$, then   $oc_{P_X}(v)=0$, as $P'$ is a $(3,t+1)$-$\{P_3\}$-membership.

If there is no $v\in V$ with $oc_{P_X}(v)=0$, then we quit this loop.
Otherwise, there is some $v\in V$ with $oc_{P_X}(v)=0$.
By maximality of $P'$, this means that $oc_{P_G}(v)=t+1$.
Pick some $p_v\in P_G$ that contains $v$.  

%If there is no vertex $v$ on any path $p\in P_G$  with $oc_{P_G}(v)\leq s$ but $oc_{P'_G}(v)=s+1$ (and hence $oc_{P_X}(v)=0$) for any $0\leq s\leq t$, we quit this loop. Otherwise, there is some $p\in P_G$ and some  $0\leq s\leq t$ 
%for which we can 
%Due to maximality of $P_G$, 
%For each $p\in P_G$, %there must be
%choose some vertex $v(p)$ occurring on $p$ with $oc_{P_G}(v(p))\leq s$ 
%\textcolor{blue}{Question 2. I think it can exists a path on $P_G$ where each vertex of the path is contained in at most $t$ paths from $P'$ then it doesn't exists the vertex $v(p)$. Path A in the image.}
%but $oc_{P'_G}(v(p))=s+1$ and hence $oc_{P_X}(v(p))=0$. 
%If such a vertex does not exist, none of the modifications described in the following needs to be done.
%
%\textcolor{blue}{Question 3. I think it may exists a path on $P_G$ with $oc_{P_X}(v(p))=1$.}

Assume that $v=v_{i+j}$ 
for some $i,j$ with ${{i\pmod{t+1}}\equiv 0}$ and $0\leq j\leq t$.
Delete $p_v$ from $P'$ (and hence from %$P_G'$ and possibly from 
$P_G$) and replace it by $v_{i+j}u_{2i/(t+1)}u_{{2i/(t+1)}+1}$.
As (previously) $oc_{P_X}(v)=0$,
the replacement was not contained in $P'$ before. 
This new $(3,t+1)$-$\{P_3\}$-membership of $G'$, called again $P'$ for the sake of simplicity, is hence not smaller than the previous one, and we can recurse, re-starting at (*); the new $P'$ decomposes again into $P_G$ and into $P_X$. 

Finally, we  arrive at a $(3,t+1)$-$\{P_3\}$-membership $P'$ of 
$G'$ such that for each vertex $v\in V$, $oc_{P_X}(v)>0$, i.e., $oc_{P_X}(v)=1$. 
As clearly $t+1\geq oc_{P'}(v)=oc_{P_X}(v)+oc_{P_G}(v)$,
this implies that $P_G$ is a $(3,t)$-$\{P_3\}$-membership of $G$.

(See Figure \ref{recursion})

%on any path from $P_G'$, $oc_{P_G}(v)=oc_{P_G'}(v)$, so that 
%$P_G=P_G'$. 
Hence, $P'$  decomposes into a $(3,t)$-$\{P_3\}$-membership $P_G$ of $G$ and a set of paths $P_X$  containing  exactly one vertex from~$G$. If $|P'|\geq n+k$, then $|P_G|=|P'|-|P_X|\geq n+k-|P_X|\geq k$
as required.
\qed
\end{proof}

\begin{figure}[htb]     
     \centerline{{\includegraphics[scale=0.60]{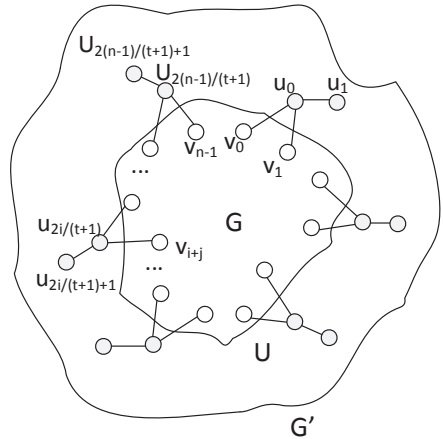}}}
     \caption{An illustration of the construction of $G'$ from $G$ and $U$ in proof of Theorem \ref{thm-P_3-membership}. In this example $t=1$. Edges of $G$ were omitted for clarity. } \label{construction}
\end{figure}

\begin{figure}[htb]     
     \centerline{{\includegraphics[scale=0.60]{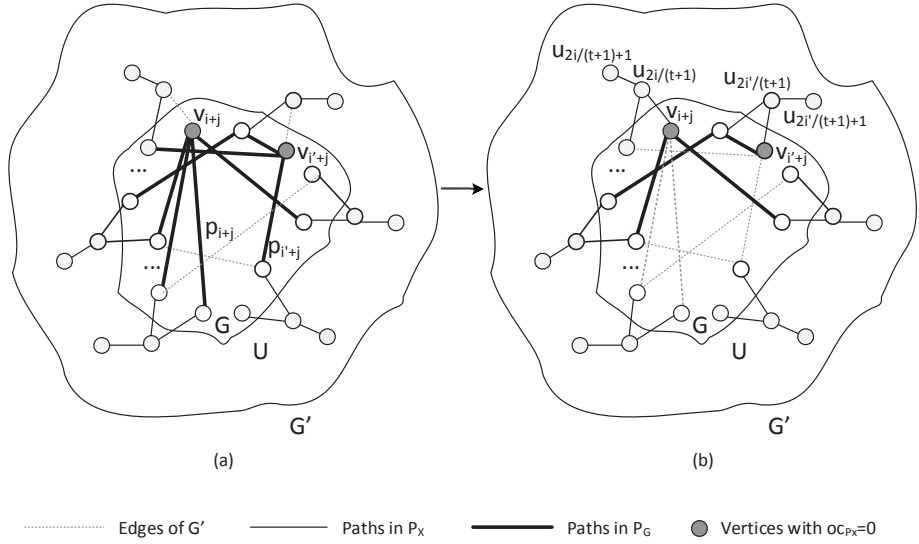}}}
     \caption{A $(3,t+1)$-$\{P_3\}$-membership $P'$ of $G'$ (for $t=1$) in proof of Theorem \ref{thm-P_3-membership}. Paths in $P_X$ and $P_G$ are indicated in black and bold black, respectively. Edges of $G'$ are indicated with dashed gray lines. In this example, there are two vertices $v_{i+j}$ and $v_{i'+j}$ with $oc_{P_X}(v_{i+j})=0$ and $oc_{P_X}(v_{i'+j})=0$. Arbitrarily, the paths $p_{{i+j}}$ and  $p_{_{i'+j}}$ are removed from $P'$ and replaced by $v_{i+j}u_{2i/(t+1)}u_{{2i/(t+1)}+1}$ and $v_{i'+j}u_{2i'/(t+1)}u_{{2i'/(t+1)}+1}$, respectively (see (b)). Observe that $P_{G}$ is now a $(3,t)$-$\{P_3\}$-membership of $G$.} \label{recursion}
\end{figure}

In fact, a similar reduction can be constructed for each $\{H\}$-Packing with $t$-Membership problem for each $H$ containing at least three vertices based on results of Kirkpatrick and Hell~\cite{Kirk78}.

%\textcolor{blue}{
%I believe that some of the technical details that are omitted in the similar reduction are something like these:}

%Let $H'$ be the induced subgraph $H'=H[V(H) \backslash u]$ for some vertex $u \in H$. For example, if $H$ is a $K_5$, $H'$ is a $K_4$, or if $H$ is a $S_5$ (star of five vertices), $H'$ will be either a $S_4$ or an independent set of 4 vertices.

%1. To construct $G'$ we introduce a set $U$ that contains $\frac{n}{t+1}$ vertex-disjoint $H'$'s.

%2. The connectivity from $G$ and $U$ is similar to the one of Theorem 1. Connect each vertex $u$ of $G$ to one $H' \in U$, in such a way that $u \cup V(H')$ forms an $H$. Each $H'$ of $U$ is connected to a set of $t+1$ different vertices of $G$ (note that these $t+1$ set of vertices connect to only that $H$).

%The proof will follow very similar. What also needs more changes is the first alteration of the $t+1$-Membership $P'$ to achieve the property that $P'$ decomposes into two sets $P_X$: a set of $H$'s using one vertex of $G$ and a subgraph of $U$, and $P_G$: a set of $H$'s using only vertices from $G$.

%The rest as you say looks like follows for the maximality of $P'$.

%Namely, pick any vertex in $H$ of minimum degree and replicate it $t+1$
%times, identifying these then with vertices $v_i,\dots,v_{i+t}$.

Omitting some technical details, we can hence state:

\begin{theorem}\label{thm-vertex-membership-hardness}
Let $H$ be a connected graph with at least three vertices.
For all $t\geq 1$, the $\{H\}$-Packing with $t$-Membership problem is NP-complete.
\end{theorem}

In particular, this applies to packings with complete graphs, with $K_3=C_3$ being the simplest example. Since in this case, the vertex-induced subgraph version coincides with the subgraph version, we can conclude:

\begin{corollary}
For each $r\geq 3$, there exists a connected graph $H$ on $r$ vertices such that, for all $t\geq 1$, the Induced-$\{H\}$-Packing with $t$-Membership problem is NP-complete.
\end{corollary}

As also the $r$-Set Packing problem is known to be NP-complete for any $r\geq 3$, we can adapt the construction of Theorem~\ref{thm-P_3-membership} also easily to obtain the following result:

\begin{theorem}
For all $t\geq 1$ and $r\geq 3$, the $r$-Set Packing with $t$-Membership problem is NP-complete.
\end{theorem}

On the positive side, we next show that the $\{P_2\}$-Packing with $t$-Membership problem can be solved in polynomial time. 

A $(k,2,t)$-$\{P_2\}$-membership $P$ of a graph $G$ is a subset of at least $k$ edges of $E(G)$ such that every vertex of $G$ is contained in at most $t$ edges of $P$. Let us denote as $G^*$ the subgraph of $G$ induced by $P$, i.e., $G^*=(V_{G^*},P)$ where $V_{G^*}=V(P)$ (the end-points of the edges in $P$). 

\begin{lemma}\label{MaxDegree}
The graph $G^*$ has maximum degree $t$.
\end{lemma}

Let $b : V (G) \to \mathbb{N}$ a degree constraint for every vertex. The problem of finding a subgraph $G^*$ of $G$ such that each vertex $v \in V(G^*)$ has degree at most $b(v)$ in $G^*$, i.e., $deg_{G^*}(v) \leq b(v)$ and the number of edges in $G^*$ is maximized is known as the \emph{degree-constrained subgraph problem} \cite{Shiloach81}. Let us refer to $G^*$ as a \emph{degree-constrained subgraph}. Y. Shiloach \cite{Shiloach81} constructs a graph $G'$ from $G$ and shows that $G$ has a degree-constrained subgraph with $k$ edges if and only if $G'$ has a maximum matching of size $|E(G)|+k$. 

By Lemma \ref{MaxDegree}, we can find a $(k,2,t)$-$\{P_2\}$-membership $P$ of $G$, by solving the degree-constrained subgraph problem with $b(v)=t$, for all $v \in V(G)$. Having a closer look at Shiloach's construction of $G'$, we observe that $|V(G')|=2|E(G)|+t|V(G)|$ and that $|E(G')|=2t|E(G)|+t|E(G)|$=$3t|E(G)|$. Thus, the maximum matching can be solved in $O(\sqrt{2|E(G)|+t|V(G)|}3t|E(G)|)$ by running Micali and Vazirani's algorithm~\cite{MicaliVazirani}. Hence, we can state:

%\textcolor{blue}{Details of this running time can be found in $Files/Membership_P2.tex$}

\begin{corollary}\label{cor-vertex-membership-easy} Let $t \geq 1$. 
$P_2$-Packing with $t$-Membership can be solved in  time
that is polynomial both in the size of the input graph $G$ and in $t$.
\\ 
%$O(\sqrt{2|E(G)|+t|V(G)|}3t|E(G)|)$ running time for all .
\end{corollary}

We can summarize Theorem~\ref{thm-vertex-membership-hardness} and Corollary~\ref{cor-vertex-membership-easy} by stating the following dichotomy result that is completely analogous to the classical one due to Kirkpatrick and Hell~\cite{Kirk78}.

\begin{theorem} (Dichotomy Theorem) 
Let $t\geq 1$.
Assuming that P is not equal to NP, then 
the $\{H\}$-Packing with $t$-Membership problem
can be solved in polynomial time if and only if $|V(H)|\leq 2$.
\end{theorem}

We now turn to the edge version of the $\mathcal{H}$-Packing with $t$-Membership problem. 
%In contrast to the vertex case, even 
%in the case of $1$-Membership, 
%the structure of this type of problem is largely unknown.
%Therefore
Notice that it is known that for each connected graph $H$ on at least three edges, the problem of finding an edge-disjoint packing of size $k$ in a graph $G$ is NP-complete, as finally shown by Dor and Tarsi~\cite{DorTar97}. 

Again, we only present one concrete hardness result,
assuming that the 
%but we think that  the 
more general hardness result follows in a similar fashion.
%\marginpar{a weak formulation, I know ...}
%together with a kind of relative hardness result.

\begin{theorem}
For any $t\geq 1$, the $\{C_3\}$-Packing with $t$-Edge Membership problem is NP-complete.
\end{theorem}

\begin{proof}
Membership in NP is obvious.
For the hardness, we proceed again by induction.
For $t=1$, the result was shown by Holyer~\cite{Hol81}.

Assume that it is known that the $\{C_3\}$-Packing with $t$-Edge Membership problem is NP-hard for some concrete $t\geq 1$. We are going to show then that the $\{C_3\}$-Packing with $(t+1)$-Edge Membership problem is NP-hard. 

Let $G=(V,E)$ and $k$ represent an instance of the $\{C_3\}$-Edge Packing with $t$-Membership problem.
We are going to construct an instance $(G',k')$ of the $\{C_3\}$-Packing with $(t+1)$-Membership problem as follows:

\begin{itemize}
\item $G'=(V',E')$, where $V'=V\cup E$ and 
$$E'=E\cup\{ve:v\in V, e\in E, v\in e\}\,,$$
\item $k'=k+|E|$.
\end{itemize}

Namely, if $P$ is a $(k,3,t)$-edge-$\{C_3\}$-membership of $G$, we can get a 
$(k',3,t+1)$-edge-$\{C_3\}$-membership of $G'$ by adding all triangles of the form $\{uv,u,v\}$ to $P$, where $u,v\in V$ and $uv\in E$.

Conversely, let $P'$ be a $(k',3,t+1)$-edge-$\{C_3\}$-membership of $G'$.
Observe that (*) the edges $(u,uv)\in E'$ with $u,v\in V$ and $uv\in E$ can only be made use of in the triangle $\{uv,u,v\}$. Moreover, any $(k',t+1)$-edge-$\{C_3\}$-membership $P'$ that does not use this triangle could be changed in a $(3,t+1)$-edge-$\{C_3\}$-membership with the same number of triangles (or more) by adding this triangle, possibly at the expense of deleting another triangle using the edge $uv$. Hence, we can assume that all $|E|$ triangles of the form  $\{uv,u,v\}$
are in the $(k',3,t+1)$-edge-$\{C_3\}$-membership of $G'$ that we consider. 
Every edge of $G$ is used exactly once in this way.
Thanks to (*), the remaining $k'-|E|=k$ triangles in $P'$ form a $(k,t)$-edge-$\{C_3\}$-membership of $G$. 
\qed
\end{proof}

It seems to be not so clear how to generalize the previous theorem to other connected graphs $H$.
The difficulty of using the same reduction idea lies 
in the fact that Observation (*) might fail.

%\textbf{HF}
%For other connected graphs $H$, the question is if there is more than one possibility to use edges in the sense of Observation (*). For instance, even $P_3$-Packings might be difficult to argue.

%I would like to add more complexity results in a sort of weakened form like the following one:

For the Packing problems with $t$-Overlap, the complexity landscape is known even less. For instance, we don't know of any concrete graph $H$ such that the corresponding $\{H\}$-Packing with $t$-Overlap problem is NP-complete for all $t\geq 0$. However, the next theorem explains at least that there are NP-hard $\{H\}$-Packing with $t$-Overlap problems for each level $t\geq 0$.

\begin{theorem}For any $t\geq0$,
there exists a connected graph $H_t$
such that the $\{H_t\}$-Packing with $t$-Overlap problem is NP-complete.
\end{theorem}

%We could then leave the study of the general situation for the future.

\begin{proof} (Sketch)
According to Caprara and Rizzi~\cite{CaprRiz2002},
the problem of packing at least $k$ triangles with zero vertex overlap in a graph of maximum degree four is NP-complete.
We reduce from this problem and add, to each vertex $v$ of the triangle packing instance graph $G$, a $K_{1,5}$, call it $S^v$ that is connected to $v$ by one edge to its center $c'$.
This gives the new graph $G'$. 
Likewise, the graph $H$ is a triangle, where we add three $K_{1,5}$'s. Clearly, any triangle packing of $G$ translates into an $H$-packing of $G'$ of the same size, and vice versa. 
However, as the only vertices in
$G'$ that are of degree six are the centers $v'$ in the $S^v$ that are connected to $v$ via a bridge, these must host the degree six vertices of $H$. Hence,
there is no way of exploiting possible overlaps, as long as $t\leq 6$.
For bigger $t$, we can change the construction by adding bigger stars $K_{1,t-1}$ to each vertex.
This finally yields the claim.
\qed
\end{proof}

This type of construction can be used in any situation where some packing problem is still NP-hard on bounded degree graphs.The trick is to add high-degree vertices that can be only matched in a specific manner.
Hence, this type of result is also true for the induced subgraph variants or for the edge variants.

More generally speaking, all our hardness constructions also show that if some packing problem is still NP-hard on planar bounded-degree graphs, this is also true for the new packing problem variants (with overlaps in any sense).

%% file: Membership_SetPacking_ToSP.tex
%\vspace{0.25cm}
%\fbox{
%\parbox{11.3cm}{
%\textbf{The $r$-Set Packing with $t$-Membership problem}
    
%    \noindent \emph{Input}: A collection $\mathcal{S}$ of distinct sets, each of size at most $r$, drawn from a universe $\mathcal{U}$ of size $n$, and a non-negative integer $k$.

%    \noindent \emph{Parameter}: $k$

%    \noindent \emph{Question}:  Does $\mathcal{S}$ contain a \emph{$(k,t)$-Set Membership}, i.e., at least $k$ sets $\mathcal{K}=\{S_1,\dots , S_k\}$ where each element of $\mathcal{U}$ is an at most $t$ sets of $\mathcal{K}$% where $t \geq 1$
%    ?}}
%\medskip

%Notice that when $t=1$, we are back to the classical $r$-Set Packing problem.

We create an instance for the $(r+1)$-Set Packing problem, (a universe $\mathcal{U}^T$ and a collection $\mathcal{S}^T$) using an instance of the $r$-Set Packing with $t$-Membership problem. 

\begin{transf}\label{MembershipSetTransf}

\textbf{Input:} $\mathcal{S}$, $\mathcal{U}$; \textbf{Output:} $\mathcal{S}^T$, $\mathcal{U}^T$, and $r$

The universe $\mathcal{U}^T$ equals $(\mathcal{U} \times\{1,\ldots,t\})\cup \mathcal{S}$.

The collection $\mathcal{S}^T$ contains all subsets of $\mathcal{U}^T$ each with at most $r+1$ elements of the following form:  $\{\{(u_1,j_1)$,$\dots,$$(u_i,j_i)$,$\dots$,$(u_{r'},j_{r'}), S\}$ $\mid S \in \mathcal{S}$, $S= \{u_1,\dots,u_{r'}\}$, for each $1 \leq j_i \leq t$ and $1 \leq i \leq r'$, where $r' \leq r\}$.
\end{transf}

Intuitively, since each element in $\mathcal{U}$ can be in at most $t$ sets of a $(k,r,t)$-set membership, we duplicate $t$ times each element. This would imply that there are $t^r$ sets in $\mathcal{S}^T$ representing the same set $S \in \mathcal{S}$. All of those sets have a common element which is precisely the set of elements in $S$, i.e., some $S \in \mathcal{S} \subset \mathcal {U}^T$. This is why the upper bound on the number of elements per subset increases.

\begin{example}
For an instance of the $r$-Set Packing with $t$-Membership with universe $\mathcal{U}=\{a,b,c,d,e,f,g,h\}$,  $\mathcal{S}=$ $\{$ $\{a,b,c,d\}$, $\{b,c,e,f\}$, $\{b,c,g,h\}$ $\}$, $k=2$ and $t=2$, the constructed instance of the ${(4+1)}$-Set Packing is as follows. 
\begin{eqnarray*}\mathcal{U}^T &=&\{(a,1),(b,1),(c,1),(d,1),(e,1),(f,1),(g,1),(h,1),
\\&&(a,2),(b,2),(c,2), (d,2),(e,2),(f,2),(g,2),(h,2)\}\\&\cup&\mathcal{S}.
\end{eqnarray*}

Some sets of the collection $\mathcal{S}^T$ are 
$\{(a,1),$ $(b,1),$ $(c,1),$ $(d,1),$ $\{a,b,c,d\}\}$, 
$\{(a,1),$ $(b,1),$ $(c,1),$ $(d,2),$ $\{a,b,c,d\}\}$, 
$\{(a,2),$ $(b,2),$ $(c,2),$ $(d,2),$ $\{a,b,c,d\}\}$, 
$\{(b,1),$ $(c,2),$ $(e,1),$ $(f,2),$ $\{b,c,e,f\}\}$, 
$\{(b,1),$ $(c,1),$ $(g,2),$ $(h,2),$ $\{b,c,g,h\}\}$.

A $(2,4+1,0)$-set packing is $\{$$\{(a,1),(b,1),(c,1),(d,1),\{a,b,c,d\}\}$, \\
$\{(b,2),(c,2),(e,1),(f,1),\{b,c,e,f\}\}$$\}$. This corresponds to the $(2,4,2)$-set packing $\{\{a,b,c,d\}$, $\{b,c,e,f\}\}$. 
\end{example}

The size of $\mathcal{U}^T$ is bounded by $|\mathcal{U}| \cdot t + |\mathcal{S}|$ $<$ $tn + r n^r =O(n^r)$. Each set in $\mathcal{S}^T$ has size at most $r+1$. For each $S \in \mathcal{S}$, we can form at most $t^r$ sets with the $tr$ ordered pairs from the elements in $S$. In this way, for each $S \in \mathcal{S}$ there are $t^r$ sets in $\mathcal{S}^T$, and $|\mathcal{S}^T|=t^r|\mathcal{S}|=O(t^r n^r)$. This lead us to the following result.

\begin{lemma}
Transformation \ref{MembershipSetTransf} can be computed in $O(t^r n^r)$ time.
\end{lemma}

Note that the parameter $k$ stays the same in this transformation, and $t$ only influences the running time of the whole construction, as the $(r+1)$-Set Packing instance will grow if $t$ gets bigger.

\begin{lemma}\label{transformationMembershipSet}
$\mathcal{S}$ has a $(k,r,t)$-set membership if and only if $\mathcal{S}^T$ has $k$ disjoint sets (i.e., a $(k,r+1,0)$-set packing). 
\end{lemma}

\begin{proof}
We build a $(k,r+1,0)$-set packing $\mathcal{K_S}$ from a $(k,r,t)$-set membership $\mathcal{K}$. For each $S_j \in \mathcal{K}$, we add a set $S^T_j$ to $\mathcal{S}^T$ with $|S_j|+1 \leq r+1$ elements.  By our construction, $S^T_j$ has one element from $\mathcal{S}$ and at most $r$ ordered pairs. 
%The element $S\in\mathcal{S}$ in $S^T_j$ equals~$S_j$. 
As $S_j \in \mathcal{S} \subset \mathcal{U}^T$,
we can choose $S^T_j\cap \mathcal{S}=\{S_j\}$. 
For each element $u_i \in S_j$ ($1 \leq i \leq |S_j| \leq r$), we add an ordered pair $(u_i,l+1)$ to $S_j$ where $l$ corresponds to the number of sets in $\{S_1, \dots, S_{j-1}\} \subset \mathcal{K}$ that contain $u_i$. Since each element of $\mathcal{U}$ is a member of at most $t$ sets of $\mathcal{K}$ (otherwise $\mathcal{K}$ would not be a solution), $l$ is at most $t-1$ and $(u_i,l+1)$ always exists in $\mathcal{U}^T$. Thus, $S^T_i \in \mathcal{S}^T$. It remains to show that the sets in $\mathcal{K_S}$ are pairwise disjoint. 
%Since $S_i \neq S_j$ for each pair $S_i,S_j \in \mathcal{K}$
By construction, none of the sets in $\mathcal{K_S}$ will share the element $S$. In addition, no pair $S^T_i,S^T_j$ shares an ordered pair $(u,l)$ for some $u \in \mathcal{U}$ and $1 \leq l \leq t$. Namely, if $S^T_i$ and $S^T_j$ share a pair $(u,l)$, this would imply that there are two sets $S_i,S_j \in \mathcal{K}$ that share the same element $u$. Without lost of generality, assume that $S_i$ appears before $S_j$ in $\mathcal{K}$. If $(u,l)$ is contained in both $S^T_i$ and $S^T_j$ then $u$ is a member of $l-1$ sets from both $\{\dots, S_i, \dots \}$ and $\{\dots, S_i, \dots, S_j \}$. This is a contradiction, since both $S_i$ and $S_j$ contain $u$. 

We construct a $(k,r,t)$-set membership $\mathcal{K}$ using a $(k,r+1,0)$-set packing $\mathcal{K_S}$. Each set $S^T_i \in \mathcal{K}_S$ has an element $S \in \mathcal{S}$. Let us denote that element as $S_i$. For each set $S^T_i \in \mathcal{K}_S$, we add $S_i$ to $\mathcal{K}$. Since the sets in $\mathcal{K}_S$ are pairwise disjoint, $S_i \neq S_j$. Now, we need to show that each element of $\mathcal{U}$ is a member of at most $t$ sets in $\mathcal{K}$. Assume otherwise for the sake of contradiction. If two sets in $\mathcal{K}$ share an element $u$, then there are two sets in $\mathcal{K}_{S}$ each one with a pair $(u,i)$ and $(u,j)$, respectively. Since the sets in $\mathcal{K}_S$ are disjoint, $i \neq j$. If an element $u \in \mathcal{U}$ is a member of more than $t$ sets in $\mathcal{K}$, then there are at least $t+1$ ordered pairs $(u,1),\dots,(u,t+1)$ contained in $t+1$ different sets in $\mathcal{K}_S$. However, by our construction of $\mathcal{U}^T$, there are at most $t$ ordered pairs that contain $u$. \qed 
\end{proof}

Then, we run the currently best kernelization algorithm for the $(r+1)$-Set Packing problem~\cite{Faisal10}. This algorithm would leave us with a new universe $\mathcal{U'}^{T}$ with at most $2r!((r+1)k-1)^{r}$ elements, as well with a collection $\mathcal{S'}^{T}$ of subsets. The following property is borrowed from \cite{Faisal10}. 

%$2(r-1)!(rk-1)^{r-1}$ elements

\begin{lemma}\label{kernelMembershipSet}
 $\mathcal{S}^T$ has a $(k,r+1,0)$-set packing if and only if $\mathcal{S'}^{T}$ has a $(k,r+1,0)$-set packing.
\end{lemma}

Next, we construct the reduced universe $\mathcal{U}'$ and $\mathcal{S'}$ using $\mathcal{U'}^{T}$ and $\mathcal{S'}^{T}$ as follows.

\begin{transf}\label{reinterpretSetMembership}

\textbf{Input:} $\mathcal{S}'^T$, $\mathcal{U}'^T$; \textbf{Output:} $\mathcal{S}'$, $\mathcal{U}'$

The reduced universe $\mathcal{U'}$ contains each element $u$ appearing in $\mathcal{U'}^{T}$. In each set of $\mathcal{S'}^T$, there is an element $S$ that will correspond to a set of $\mathcal{S'}$. 
\end{transf}

$\mathcal{U}'$ and $\mathcal{S}'$ together with $k$, give the reduced $r$-Set Packing with $t$-Membership instance we are looking for. By our construction, $\mathcal{U}'$ will have at most $2r!((r+1)k-1)^{r}$  elements. This reduction property allows us to state:

\begin{theorem}
The $r$-Set Packing with $t$-Membership has a problem kernel with $O((r+1)^r k^{r})$ elements from the given universe.
\end{theorem}

Our kernelization algorithm is summarized in Algorithm \ref{SetMembershipKernel}.

\begin{algorithm}[tbh] 
  \caption{Kernelization Algorithm - $r$-Set Packing with $t$-Membership} 
	\textbf{Input:} $\mathcal{U}$ and $\mathcal{S}$ \\
	\textbf{Output:} Reduced $\mathcal{U'}$ and $\mathcal{S'}$
	\begin{algorithmic}[1] 	
	
		\STATE{Construct a universe $\mathcal{U}^T$ and collection $\mathcal{S}^T$ using Transformation \ref{MembershipSetTransf}}
		
		\STATE{Run the $(r+1)$-Set Packing kernelization algorithm~\cite{Faisal10} on $\mathcal{U}^T$ and $\mathcal{S}^T$}
	  
		\STATE{Reinterpret the kernel $\mathcal{S}'^T$, $\mathcal{U}'^T$ using Transformation \ref{reinterpretSetMembership}.}
				
		\STATE{Return the output $\mathcal{U}'$, $\mathcal{S}'$ of Transformation \ref{reinterpretSetMembership}.}
		
		\end{algorithmic}\label{SetMembershipKernel}
\end{algorithm}

%% file: Membership_Vertex_ToMembershipSP.tex
%\vspace{0.25cm}
%\fbox{
%\parbox{11.3cm}{
%\textbf{The $\mathcal{H}$-Packing with $t$-Membership problem}.
%    
%    \noindent \emph{Input}: A graph $G$, and a non-negative integer $k$.
%
%    \noindent \emph{Parameter}: $k$
%
%    \noindent \emph{Question}: Does $G$ contain a $(k,t)$-$\mathcal{H}$-\emph{Membership}, i.e.,  a set of at least $k$ subgraphs $\mathcal{K}=\{H_1, \dots ,H_k\}$ where each $H_i$ is isomorphic to some graph $H \in \mathcal{H}$, $V(H_i) \neq V(H_j)$, and every vertex in $V(G)$ is contained in at most $t$ %$H$-
%    subgraphs of~$\mathcal{K}$%  where $t \geq 1$
%    ?}}
%\vspace{0.25cm}

Note that there could exists in $G$ more than one $\mathcal{H}$-subgraph with the same set of vertices (but different set of edges). However, we claim next that only one of those $\mathcal{H}$-subgraphs can be in a solution.

\begin{lemma}\label{SameVerticesMembership}
Let $H_i$ and $H_j$ be a pair of $\mathcal{H}$-subgraphs in $G$  such that $V(H_i)=V(H_j)$ but $E(H_i) \neq E(H_j)$. Any $(k,r,t)$-$\mathcal{H}$-membership of $G$ that contains $H_i$ does not contain $H_j$ (and vice versa). Furthermore, we can replace $H_i$ by $H_j$ in such membership.
\end{lemma}

We denote as $\mathcal{H}_{G}$ the set of all $\mathcal{H}$-subgraphs in $G$; thus, $|\mathcal{H}_G|=O(|\mathcal{H}|n^{{r(\mathcal{H})}^2})$. We can find a $(k,r,t)$-$\mathcal{H}$-membership from $G$ by selecting $k$ $\mathcal{H}$-subgraphs from $\mathcal{H}_G$ such that every vertex of $V(G)$ is contained in at most $t$ of those subgraphs. By Lemma \ref{SameVerticesMembership}, we can apply the next reduction rule to $\mathcal{H}_G$.

\begin{redrule}\label{cleanupRepeated}
For any pair of $\mathcal{H}$-subgraphs $H_1,H_2$ in $\mathcal{H}_G$ such that $V(H_1)=V(H_2)$, we arbitrary select one and remove the other from $\mathcal{H}_G$.
\end{redrule}

It is important to clarify that we only apply Reduction Rule \ref{cleanupRepeated} to $\mathcal{H}_G$ for the vertex and induced version of this problem. Thus, after applying this rule, $|\mathcal{H}_G|=O(|\mathcal{H}|n^{r(\mathcal{H})})$.

We construct an instance for the $r$-Set Packing with $t$-Membership as follows.

\begin{transf}\label{VertexPackingToOSP}
\textbf{Input:} $G$, $\mathcal{H}$; \textbf{Output:} $\mathcal{U}$, $\mathcal{S}$, and $r$

The universe $\mathcal{U}$ equals $V(G)$.

There is a set in $\mathcal{S}$ for each $\mathcal{H}$-subgraph $H$ in $\mathcal{H}_G$ and $S=V(H)$.

Furthermore, let $r=r(\mathcal{H})$. 
\end{transf}

%\textbf{Is there any difference between  $\mathcal{H}_G$ and  $\mathcal{H}$?
%What happens if there are, e.g., two non-isomorphic graphs with the same set of vertices
%in $\mathcal{H}$? Are "complete overlaps" of some $P_3$ and some $C_3$ on the same vertices
%allowed or disallowed?
%}

%\textcolor{blue}{$\mathcal{H}$ is the family of graphs we are looking in $G$, while $\mathcal{H}_G$ is the set of all subgraphs of $G$ (the ``copies'' in $G$) each isomorphic to some $H \in \mathcal{H}$. After trying all possible combinations of $n^r$ (actually is for every size of $H \in \mathcal{H}_G$) vertices we ended-up with this list of subgraphs (copies of some $H \in \mathcal{H}$) now we have to select a $(k,t)$-$\mathcal{H}$-\emph{Membership} from that list.}

%\textcolor{blue}{That depends.In the vertex-version the condition says no. We don't allow complete overlaps (well sort of). I think that when we transform a copy of $P_3$ and a copy of $C_3$ in $G$ to sets in $\mathcal{S}$, they will be two identical sets (this is why we apply Rule 1). However, for the edge-version only our second transformation will avoid the complete overlap problem.
%Note however that a copy contained in other copy is not avoided (for example a $P_3$ contained in a $P_6$). we can have up to $t$ times this situation (because the overlap restriction). }

In this way, $|\mathcal{U}|=O(n)$ and $|\mathcal{S}| = |\mathcal{H}_G| = O(|\mathcal{H}|n^{r(\mathcal{H})})$. Each set in $\mathcal{S}$ has size at most $r(\mathcal{H})$. 

\begin{lemma}
$G$ has a $(k,r,t)$-$\mathcal{H}$-membership if and only if $\mathcal{S}$ has a $(k,r,t)$-set membership
\end{lemma}

\begin{proof}
We build a $(k,r,t)$-set membership $\mathcal{K}_S$ from a $(k,r,t)$-$\mathcal{H}$-membership $\mathcal{K}$. For each $\mathcal{H}$-subgraph $H_i$ in $\mathcal{K}$, we add a set $S_i = V(H_i)$ to $\mathcal{K}_S$. By our construction, $S_i \in \mathcal{S}$. Given that every vertex of $\mathcal{U}$ is contained in at most $t$ 
$\mathcal{H}$-subgraphs of $\mathcal{K}$ then each element of $\mathcal{U}$ will be contained in at most $t$ sets of $\mathcal{K}_S$.

Given a $(k,r,t)$-set membership $\mathcal{K}_S$, we build a $(k,r,t)$-$\mathcal{H}$-membership $\mathcal{K}$ of $G$. For each set $S_i$ in $\mathcal{K}_S$ we add an $\mathcal{H}$-subgraph $H_i \subseteq G[S_i]$. Since, $\mathcal{U} = V(G)$, $H_i$ is an $\mathcal{H}$-subgraph of $G$. Each vertex of $G$ will be contained in at most $t$ $\mathcal{H}$-subgraphs of $\mathcal{K}$; otherwise there would be one element of $\mathcal{U}$ contained in more than $t$ sets of $\mathcal{K}_S$. \qed
\end{proof}

We obtain reduced universe $\mathcal{U'}$ and $\mathcal{S'}$ for the constructed instance of the $r$-Set Packing with $t$-Membership with Algorithm \ref{SetMembershipKernel}. After that, the reduced graph for the original instance is obtained as follows.

\begin{transf}\label{FromUniverseToGraph}
\textbf{Input:} $G$, $\mathcal{U'}$; \textbf{Output:} $G'$

Return $G' = G[\mathcal{U'}]$. 
\end{transf}

%We obtain reduced universe $\mathcal{U'}$ and $\mathcal{S'}$ for the constructed instance of the Set Packing with $t$-Membership as showed in Subsection \ref{membershipSets}. After that, the reduced graph for the original instance is $G' = G[\mathcal{U'}]$. 

\begin{lemma}
$G'$ has a $(k,r,t)$-$\mathcal{H}$-membership if and only if $\mathcal{S'}$ has a $(k,r,t)$-set membership.
\end{lemma}

Since $|\mathcal{U'}|=O((r+1)^rk^r)$ (by Algorithm \ref{SetMembershipKernel}), this reduction property allows us to state:

\begin{theorem}\label{VertexMembershipKernel}
$\mathcal{H}$-Packing with $t$-Membership has a problem kernel with $O((r+1)^r k^{r})$ vertices.
\end{theorem}

The induced version of the problem seeks for at least $k$ induced $\mathcal{H}$-subgraphs in $G$ where each vertex in $V(G)$ is contained in at most $t$ of these $\mathcal{H}$-subgraphs. To achieve a problem kernel, we redefine $\mathcal{S}$ in Transformation \ref{VertexPackingToOSP}, adding a set $S=V(H)$ per each each \textbf{induced} $\mathcal{H}$-subgraph $H$ of $G$.

\begin{theorem}
Induced-$\mathcal{H}$-Packing with $t$-Membership has a problem kernel with $O((r+1)^r k^{r})$ vertices.
\end{theorem}

Given that all our kernelization algorithms for all our graph problems with $t$-Membership and $t$-Overlap use a similar methodology, we will provide a common kernelization algorithm that can be used for all of them. This algorithm besides receiving as input $G$, $\mathcal{H}$ and $k$, it also has four additional parameters: ``Problem ${\pi}$'', [Transformation], [KernelAlg], and [Reinterpretation]. The last three parameters are placeholders for specific routines that vary according to the type of problem that we are reducing. For example, the kernelization algorithm that we described above corresponds to Algorithm \ref{GenericKernelization} with input $G$,$\mathcal{H}$,$k$, ``$r$-SetPacking with $t$-Membership'', Transformation \ref{VertexPackingToOSP}, Algorithm \ref{SetMembershipKernel}, and Transformation \ref{FromUniverseToGraph}. Notice that the parameter $k$ stays the same, as it is not changed in any of our routines. 

\begin{algorithm}[tbh] 
  \caption{Kernelization Algorithm - $\mathcal{H}$-Packing with Membership/Overlap} 
	\textbf{Input:} $G$, $\mathcal{H}$, $k$, ``Problem ${\pi}$'', [Transformation], [KernelAlg], and [Reinterpretation]. \\
	\textbf{Output:} $G'$
	\begin{algorithmic}[1] 	
	
		\STATE{Construct an instance of [Problem ${\pi}$] (a universe $\mathcal{U}$ and a collection $\mathcal{S}$) by running [Transformation($G$,$\mathcal{H}$)]}
		
		\STATE{Obtain reduced $\mathcal{U'}$ and $\mathcal{S'}$ by running [KernelAlg($\mathcal{U}$,$\mathcal{S}$,$k$)] (i.e., a kernelization algorithm of [Problem ${\pi}$])}
	  
		\STATE{Reinterpret $G'$ by running [Reinterpretation$(G,\mathcal{U'})$]}
				
		\STATE{Return $G'$}
		
		\end{algorithmic}\label{GenericKernelization}
\end{algorithm} 

\subsubsection{$\mathcal{H}$-subgraphs with identical set of vertices - Vertex-Membership.}

\input{Membership_Vertex_ToSP_v2}

%% file: Membership_Vertex_ToSP_v2.tex
We discuss a more flexible definition of the $\mathcal{H}$-Packing with $t$-Membership problem and
allow $\mathcal{H}$-subgraphs in a $(k,r,t)$-$\mathcal{H}$-membership with identical sets of vertices (but require different sets of edges). This is equivalent to remove the condition $V(H_i) \neq V(H_j)$ from  the $\mathcal{H}$-Packing with $t$-Membership problem definition. Let us indicate this variant by adding ISV (identical set of vertices) to the problem name. 
We present a kernelization algorithm through some PPT reduction for this version, as well.

Take for example the input graph of Figure \ref{membership2} and suppose that $\mathcal{H}=\{C_4,K_4\}$ and $t=2$. With the problem definition that asks for $V(H_i) \neq V(H_j)$ in a $(k,r,t)$-$\mathcal{H}$-membership, either the $K_4$ with vertices $\{b,e,c,f\}$ or one of the two $C_4$'s composed of the vertices $\{b,e,c,f\}$ will be part of a solution. However, it may exist the case that we still want to have both $\mathcal{H}$-subgraphs in a solution (up to $t$ times because the overlap condition). To achieve the flexibility of having $\mathcal{H}$-subgraphs with identical set of vertices in a $(k,r,t)$-$\mathcal{H}$-membership, we transform an instance of $\mathcal{H}$-Packing with $t$-Membership to an instance of the ($r+1$)-Set Packing problem (element disjoint) rather to an instance of $r$-Set Packing with $t$-Membership.

\begin{figure}[htb]     
     \centerline{{\includegraphics[scale=0.60]{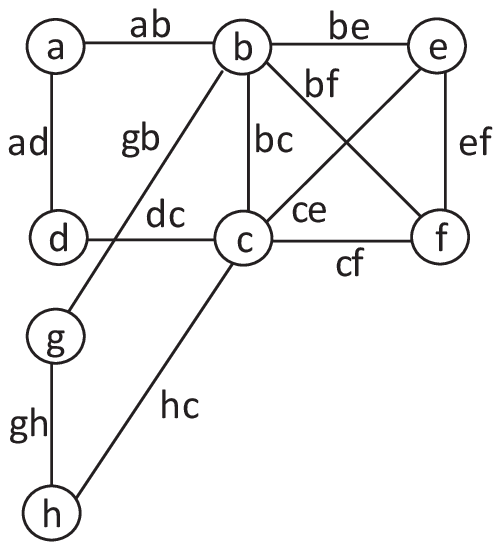}}}
     \caption{An input graph $G$ for the $\mathcal{H}$-Packing with $t$-Membership problem.} \label{membership2}
\end{figure}

Recall that $\mathcal{H}_{G}$ is the set of all $\mathcal{H}$-subgraphs in $G$. Notice that for this transformation we do not apply Reduction Rule \ref{cleanupRepeated} to $\mathcal{H}_G$. Next, we create an instance for the $(r+1)$-Set Packing problem, (a universe $\mathcal{U}$ and a collection $\mathcal{S}$) as follows.

Let $\mathcal{E}$ defines a collection of sets of edges each of which induces a subgraph isomorphic to some $H \in \mathcal{H}$. That is, for each $E \in \mathcal{E}$, there is an $H \subseteq G[V(E)]$ that is an $\mathcal{H}$-subgraph. In this way, for each $E \in\mathcal{E}$, $|E| \leq r^2$, and $|\mathcal{E}| \leq n^{r^2}$. Slightly abusing our  notation, let $V(E)$ denote the set of vertices corresponding to the end-points of the edges $E$.  

\begin{transf}\label{AltransfForVertexMembership}
\textbf{Input:} $G$, $\mathcal{H}$;  \textbf{Output:} $\mathcal{U}$, $\mathcal{S}$, and $r$.

The universe $\mathcal{U}$ equals $(V(G)\times\{1,\dots,t\}) \cup \mathcal{E}$.

Let $r=r(\mathcal{H})$.

The collection $\mathcal{S}$ contains all subsets of $\mathcal{U}$ each with at most $r+1$ elements $\{\{(v_1,j_1)$,$\dots,$$(v_i,j_i)$,$\dots$,$(v_{r'},j_{r'}), E\}$ $\mid E \in \mathcal{E}$, $V(E)= \{v_1,\dots,v_{r'}\}$, for each $1 \leq j_i \leq t$ and $1 \leq i \leq r'$, where $r' \leq r\}$.
\end{transf}

The size of $\mathcal{U}$ is upper-bounded by $|\mathcal{U}|=|V(G)| \times t + |\mathcal{E}|$ $\leq$ $tn + n^{m(\mathcal{H})} =O(n^{r^2})$. 
Each set in $\mathcal{S}$ has size at most $r+1$ (where $r=r(\mathcal{H})$) and $|\mathcal{S}| \leq t^{r} |\mathcal{H}| n^{m(\mathcal{H})}$.

Intuitively, each set in $\mathcal{S}$ would represent an $\mathcal{H}$-subgraph in $G$. Since each vertex in $V(G)$ can be in at most $t$ $\mathcal{H}$-subgraphs, we replicate  each vertex $t$ times. To allow $\mathcal{H}$-subgraphs with the same set of vertices, we add the element $E$ to the sets. In this way, $\mathcal{H}$-subgraphs with identical set of vertices will have a different element $E$.

\begin{example}
Let us consider the graph of Figure \ref{membership2}, $\mathcal{H}=\{C_4,K_4\}$, $t=3$, and $k=3$. 

%The set
$\mathcal{E}=$$\{$$\{ab,ad,bc,dc\}$, $\{bg,bc,gh,hc\}$, $\{be,ef,cf,bc\}$, $\{be,ce,cf,bf\}$, $\{be,bc,ef,cf,bf,ce\}$ $\}$
%while the universe 
and 

$\mathcal{U} =\{(a,1),(b,1),(c,1)$, $(d,1)$, $(e,1)$, $(f,1),(g,1),(h,1)$, $(a,2),(b,2),(c,2),$ $(d,2),$  $(e,2),(f,2),(g,2),(h,2)\}$  $\cup$ $\mathcal{E}$. 

Some sets of the collection $\mathcal{S}$ are 

$\{(a,1),$ $(b,1),$ $(c,1),$ $(d,1),$ $\{ab,ad,bc,bd\}\}$, 

$\{(a,1),$ $(b,1),$ $(c,1),$ $(d,2),$ $\{ab,ad,bc,bd\}\}$, 

$\{(b,1),$ $(c,1),$ $(e,1),$ $(f,1),$ $\{be,ef,cf,bc\}\}$, 

$\{(b,2),$ $(c,2),$ $(e,2),$ $(f,2),$ $\{be,ce,cf,bf\}\}$, 

$\{(b,3),$ $(c,3),$ $(e,3),$ $(f,3),$ $\{be,bc,ef,cf,bf,ce\}\}$.

A $(3,4+1,3)$-set packing is $\{$$\{(b,1),$ $(c,1),$ $(e,1),$ $(f,1),$ $\{be,ef,cf,bc\}\}$, 
$\{(b,2),$ $(c,2),$ $(a,1),$ $(d,1),$ $\{ab,ad,bc,bd\}\}$, and  
$\{(b,3),$ $(c,3),$ $(g,3),$ $(h,3),$ $\{be,bc,ef,cf,bf,ce\}\}$$\}$.

This corresponds to the $(3,4,3)$-$\mathcal{H}$-membership  $\{H_1=(\{b,c,e,f\}$, $\{be,ef,cf,bc\})$, $H_2=(\{a,b,c,d\}$, $\{ab,ad,bc,bd\})$ and  $H_3=(\{b,c,e,f\}$, $\{be,bc,ef,cf,bf,ce\})$. 
\end{example}

\begin{lemma}\label{membershipTransformationVertex}
$G$ has a $(k,r,t)$-$\mathcal{H}$-membership (ISV) if and only if $\mathcal{S}$ has a $(k,r+1,0)$-set packing.
\end{lemma}

\begin{proof}
We build a $(k,r+1,0)$-set packing $\mathcal{K_S}$ from a $(k,r,t)$-$\mathcal{H}$-membership (ISV) $\mathcal{K}$. 
For each $H_j \in \mathcal{K}$, we add a set $S_j$ with $|V(H_j)|+1 \leq r+1$ elements as follows.  
The element $E$ in $S_j$ corresponds to the set of edges $E(H_j)=\{e_1,\dots,e_{|E|}\}$. 
Note that $E(H_j) \in \mathcal{E} \subseteq \mathcal{U}$. For each vertex $v_i \in V(H_j)$ ($1 \leq i \leq |V(H_j)| \leq r$), we add an ordered pair $(v_i,l+1)$ to $S_j$, where $l$ corresponds to the number of $\mathcal{H}$-subgraphs in $\{H_1, \dots, H_{j-1}\} \subseteq \mathcal{K}$ that $v_i$ is contained in. Since each vertex is contained in at most $t$ $\mathcal{H}$-subgraphs, $0 \leq l \leq t-1$ and $(v_i,l+1)$ always exists in $\mathcal{U}$. By our construction of $\mathcal{S}$,  $S_j \in \mathcal{S}$. It remains to show that the sets in $\mathcal{K_S}$ are pairwise disjoint. Since $E(H_i) \neq E(H_j)$ for each pair $H_i,H_j \in \mathcal{K}$, none of the sets in $\mathcal{K_S}$ will share the same element $E$.  In addition, no pair $S_i,S_j$ shares an ordered pair $(v,l)$ for some $v \in V(G)$ and $1 \leq l \leq t$. If $S_i$ and $S_j$ share a pair $(v,l)$ this would imply that there are two $\mathcal{H}$-subgraphs $H_i$ and $H_j$ that share the same vertex $v$.  Without lost of generality, assume that $H_i$ appears before $H_j$ in $\mathcal{K}$. If $(v,l)$ is contained in both $S_i$ and $S_j$  then $v$ is a member of $l-1$ $\mathcal{H}$-subgraphs from both $\{\dots, H_i, \dots \}$ and $\{\dots, H_i, \dots, H_j \}$, a contradiction, since both $H_i$ and $H_j$ contain $v$.

We construct a $(k,r,t)$-$\mathcal{H}$-membership (ISV) $\mathcal{K}$ using a $(k,r+1,0)$-set packing $\mathcal{K_S}$. 
Each set $S_i \in \mathcal{K}_S$ has an element $E_{S_i} \in \mathcal{E}$ which corresponds to a set of edges in $E(G)$ that forms a  subgraph isomorphic to $H \in \mathcal{H}$ in $G$. In addition, we take the vertex $v_l$ that appears in each $(v_l,j) \in S_i$ (for $1 \leq l \leq r$ and $j \in [1...t]$) and we add it to a set $V_{S_i}$. In this way, for each set $S_i \in \mathcal{K}_S$, we add to $\mathcal{K}$ an $\mathcal{H}$-subgraph $H_{S_i}$ with a set of vertices $V_{S_i}$ and a set of edges $E_{S_i}$, i.e., $H_{S_i}=(V_{S_i},E_{S_i})$. Since the sets in $\mathcal{K}_S$ are pairwise disjoint, $E(H_{S_i}) \neq E(H_{S_j})$.
Now, we need to show that each vertex of $V(G)$ is a member of at most $t$ $\mathcal{H}$-subgraphs in $\mathcal{K}$. Assume otherwise by contradiction. If two $\mathcal{H}$-subgraphs in $\mathcal{K}$ share a vertex $v$ then there are two sets in $\mathcal{K}_{S}$ where each one has a pair $(v,i)$ and $(v,j)$, respectively. Since the sets in $\mathcal{K}_S$ are disjoint then $i \neq j$. If a vertex $v \in V(G)$ is a member of more than $t$ $\mathcal{H}$-subgraphs in $\mathcal{K}$ then there are at least $t+1$ ordered pairs $(v,1),\dots,(v,t+1)$ contained in $t+1$ different sets in $\mathcal{K}_S$. However, by our construction of $\mathcal{U}$, there are at most $t$ ordered pairs that contain $v$. \qed 
\end{proof}

Given a reduced universe $\mathcal{U}' \subseteq \mathcal{U}$, we construct a graph $G'$ as follows.

\begin{transf}\label{FromUniverseToGraph_2}
\textbf{Input:} $G$, $\mathcal{U'}$; \textbf{Output:} $G'$

We take each vertex $v$ that appears in each $(v,i) \in\mathcal{U}'$ and the vertices in each $E \in \mathcal{E}'$ and consider the graph $G'$ that is induced by all these in $G$.
\end{transf}

\begin{lemma}
$G'$ has a $(k,r,t)$-$\mathcal{H}$-membership (ISV) if and only if $\mathcal{S'}$ has a $(k,r+1,0)$-set packing.
\end{lemma}

Our kernelization algorithm for $\mathcal{H}$-Packing with $t$-Membership problem (ISV) corresponds to Algorithm \ref{GenericKernelization} with input: $G$, $\mathcal{H}$, $k$, ``$(r+1)$-Set Packing'', Transformation \ref{AltransfForVertexMembership}, kernelization algorithm \cite{Faisal10}, and Transformation \ref{FromUniverseToGraph_2}. Given that the universe $\mathcal{U'}$ returned by algorithm \cite{Faisal10} has at most 
$2r!((r+1)k-1)^{r}$ elements, then by our construction, $G'$ has at most $2r!((r+1)k-1)^{r}$ vertices. This implies that our kernelization algorithm reduces the $\mathcal{H}$-Packing with $t$-Membership problem (ISV) to a kernel with $O((r+1)^r k^{r})$ vertices.

\begin{theorem}
The $\mathcal{H}$-Packing with $t$-Membership problem (ISV) possess a kernel with $O((r+1)^r k^{r})$ vertices.
\end{theorem}

%% file: Membership_Edge_ToMembershipSP.tex
%In the $\mathcal{H}$-Packing with $t$-Edge Membership Problem, where $t\geq 1$, we bound the number of $\mathcal{H}$-subgraphs that an edge of $E(G)$ can belong to. 

%\vspace{0.25cm}
%\fbox{
%\parbox{11.3cm}{
%\textbf{The $\mathcal{H}$-Packing with $t$-Edge Membership problem}
    
%    \noindent \emph{Input}: A graph $G$, and a positive integer $k$.

%    \noindent \emph{Parameter}: $k$

%    \noindent \emph{Question}: Does $G$ contain a $(k,t)$-Edge-$\mathcal{H}$-Membership, i.e.,  a set of at least $k$ subgraphs $\mathcal{K}=\{H_1, \dots ,H_k\}$ where each $H_i$ is isomorphic to some graph $H \in \mathcal{H}$ and every edge in $E(G)$ belongs to at most $t$ $H$-subgraphs of~$\mathcal{K}$%, where $t \geq 1$
%    ?
%}}
%\vspace{0.25cm}

Similarly as with vertex-membership, we will reduce the $\mathcal{H}$-Packing with $t$-Edge Membership Problem to the $r$-Set Packing with $t$-Membership problem.

For the edge-membership version, we assume that each $H \in \mathcal{H}$ has no isolated vertices. Otherwise, we can replace each $H$ that has a set  $H_I$ of isolated vertices with a graph $H'=H\backslash H_I$. Let $\mathcal{H}'$ denote the family of subgraphs with these modified graphs. Furthermore, let $I$ denote the set of isolated vertices contained in $G$.

\begin{lemma}\label{GNoIsolated}
G has a $(k,r,t)$-edge-$\mathcal{H}$-membership if and only if $G \backslash I$ has a $(k,r,t)$-edge-$\mathcal{H}'$-membership.
\end{lemma}

Henceforth, for the remainder of this Section we assume that each $H \in \mathcal{H}$ and hence $G$ as well do not contain isolated vertices. 

We create next an instance of the $r$-Set Packing with $t$-Membership problem using an instance of the $\mathcal{H}$-Packing with $t$-Edge Membership Problem.

\begin{transf}\label{edgePackingToOSP}
\textbf{Input:} $G$, $\mathcal{H}$;  \textbf{Output:} $\mathcal{U}$, $\mathcal{S}$ and $r$.

The universe $\mathcal{U}$ equals to $E(G)$.

There is a set $S$ in $\mathcal{S}$ for each $\mathcal{H}$-subgraph $H$ in $\mathcal{H}_G$, and $S=E(H)$. 

Furthermore, let $r=m(\mathcal{H})$.
\end{transf}

In this way, $|\mathcal{U}|=O(n^2)$ and $|\mathcal{S}| = |\mathcal{H}_G| = O(|\mathcal{H}|n^{2r})$.

%\marginpar{We have to specify $r$; it is EDGES now.} Each set in $\mathcal{S}$ has size $O(r^2)$.%
%\marginpar{\textcolor{blue}{Actually $r$ still denotes the number of vertices, but note that the size of a set in $\mathcal{S}$ is $O(r^2)$. I used the letter $m$ for denoting edges. I thought that it would be better to keep $r$ always for vertices.}}

%\textbf{We still have to specify $r$: Always having this refer to edges makes no sense, at least no sense any longer, as we are now "thinking" from the perspective of hypergraph packing / set packing. The "elements" of the universe can be vertices or edges, whatever we are modeling.}

%\textbf{A short further comment on the "difficulty" with the connectedness question.
%As I indicated in the comment to the proof below, I don't see a problem with subgraphs like
%"two parallel edges", i.e., $K_2\cup K_2$, or $2K_2$. Problems come with isolated vertices.
%I think that we can transform in polynomial time any instance that contains a graph with some, say, $I$, isolated vertices in $H$ to one without.
%Namely, as $H$ is fixed, $I$ is also fixed. If $G$ is an instance with less than $I$ vertices, then $G$ cannot host $H$. Otherwise, $G$ can host $H$ if and only if it can host $H$ with all isolates deleted.
%This is because only edges is what counts. Doing this for all graphs in $\mathcal{H}$, we can reduce to a problem where all graphs that we want to host have no isolated vertices.}

\begin{lemma}
$G$ has a $(k,r,t)$-edge-$\mathcal{H}$-membership if and only if $\mathcal{S}$ has a $(k,r,t)$-set membership.
\end{lemma}

\begin{proof}
We build a $(k,r,t)$-set membership $\mathcal{K}_S$ from a $(k,r,t)$-edge-$\mathcal{H}$-membership $\mathcal{K}$ of $G$. For each $\mathcal{H}$-subgraph $H_i$ in $\mathcal{K}$, we add a set $S_i = E(H_i)$ to $\mathcal{K}_S$. 
%Since, $H$ does not contain isolated vertices, and $\mathcal{U}=E(H)$, then $S_i$ exists in $\mathcal{S}$. 
%\textbf{In which sense? And why wouldn't this "existence" be ok if $H_i$ is NOT connected? The only problem that I can see is with isolated vertices, as they are somehow completely neglected when talking about edges only.}
Since every edge of $E(G)$ is contained in at most $t$ $\mathcal{H}$-subgraphs of $\mathcal{K}$, every element of $\mathcal{U}$ will be contained in at most $t$ $\mathcal{H}$-subgraphs of $\mathcal{K}_S$.

Given a $(k,r,t)$-set membership $\mathcal{K}_S$ we build a $(k,r,t)$-edge-$\mathcal{H}$-membership $\mathcal{K}$ of $G$. For each set $S_i$ in $\mathcal{K}_S$ we add an $\mathcal{H}$-subgraph $H_i$ with set of vertices $V(S_i)$ and the set of edges $S_i$. Each edge of $G$ will be contained in at most $t$ $\mathcal{H}$-subgraphs of $\mathcal{K}$. Otherwise, $\mathcal{K}_S$ would not be a $(k,r,t)$-set membership. \qed
\end{proof}

We obtain a reduced universe $\mathcal{U'}$ and a collection $\mathcal{S'}$ for the constructed instance of the $r$-Set Packing with $t$-Membership with Algorithm \ref{SetMembershipKernel}. The reduced graph for the original instance is obtained with the following transformation.

\begin{transf}\label{FromUniverseToGraph_v3}
\textbf{Input:} $G$, $\mathcal{U'}$; \textbf{Output:} $G'$

Return $G' = G[V(\mathcal{U'})]$. 
\end{transf}

\begin{lemma}
$G'$ has a $(k,r,t)$-edge-$\mathcal{H}$-membership if and only if $\mathcal{S'}$ has a $(k,r,t)$-set membership.
\end{lemma}

Given that $G$ does not contain isolated vertices, $|V(G')| \leq 2|\mathcal{U'}|$,  
%\textbf{This shows again that what you actually need is "no isolates".}

\begin{theorem}\label{KernelEdgeMembership}
The $\mathcal{H}$-Packing with $t$-Edge Membership problem has a kernel with $O((r+1)^{r} k^{r})$ vertices where $r=m(\mathcal{H})$.
\end{theorem}

Our kernelization algorithm for the $\mathcal{H}$-Packing with $t$-Edge Membership problem corresponds to Algorithm \ref{GenericKernelization} with input $G$, $\mathcal{H}$, $k$, ``$r$-Set Packing with $t$-Membership'', Transformation \ref{edgePackingToOSP}, Algorithm \ref{SetMembershipKernel}, and Transformation \ref{FromUniverseToGraph_v3}.

\subsubsection{$\mathcal{H}$-subgraphs with non identical set of vertices - Edge-Membership.}
%\textbf{HF} Explain the new problem first!
%Is this really about NOT being vertex-disjoint?!
%I am puzzled.

\input{Membership_Edge_ToSP}

%% file: Membership_Edge_ToSP.tex
We present a stricter definition of the $\mathcal{H}$-Packing with $t$-Edge Membership problem to
avoid $\mathcal{H}$-subgraphs in a $(k,m,t)$-edge-$\mathcal{H}$-membership with identical sets of vertices. This is equivalent to add the condition $V(H_i) \neq V(H_j)$ to the $\mathcal{H}$-Packing with $t$-Edge Membership problem definition. We will indicate this variant by adding NISV (non identical set of vertices) to the problem name. 

\begin{example}
Figure \ref{membership2} shows an input graph for the $\mathcal{H}$-Packing with $t$-Edge Membership Problem. Let us consider $\mathcal{H}=\{C_4\}$, $t=2$, and $k=2$. The constructed instance for the 
$r$-Set Packing with $t$-Membership will have
\\
$\mathcal{U}=\{ab,ad,bg,dc,gh,hc,bc,be,bf,bc,ce,cf,ef\}$ and\\ $\mathcal{S}$=$\{\{ab,ad,bc,dc\}$,$\{bg,$ $gh,hc,bc\}$, $\{be,bc,cf,ef\}$,$\{be,bf,cf,ce\}\}$
 
A $(2,4,2)$-set membership is $\{\{be,bc,cf,ef\},\{be,bf,cf,ce\}\}$. This corresponds to the $(2,4,2)$-edge-$\mathcal{H}$-membership
$H_1=(\{\{b,e,c,f\}\{be,bc,cf,ef\}),H_2=(\{\{b,e,c,f\}\{be,bf,cf,ce\})\}$.
\end{example}

Observe that $V(H_1)=V(H_2)$. This is still a solution for the $\mathcal{H}$-Packing with $t$-Edge Membership problem that does not require that $V(H_i) \neq V(H_j)$ for each pair of $\mathcal{H}$-subgraphs of a $(k,m,t)$-edge-$\mathcal{H}$-membership.

It could be of interest to exclusively have $\mathcal{H}$-subgraphs with different set of vertices in a $(k,m,t)$-edge-$\mathcal{H}$-membership. We reduce the $\mathcal{H}$-Packing with $t$-Edge Membership problem (NISV) directly to ($r+1$)-Set Packing (element-disjoint), instead of reducing to $r$-Set Packing with $t$-Membership. 

Let $\mathcal{V}$ be a collection of sets of vertices each of which forms a subgraph (not necessarily induced) isomorphic to some $H \in \mathcal{H}$. That is, for each $V \in \mathcal{V}$, there is an $H' \subseteq G[V]$ that is an $\mathcal{H}$-subgraph, (here $\subseteq$ also refers to the subgraph relation). Hence, for each $V \in\mathcal{V}$, $|V| \leq r$, and $|\mathcal{V}| = O(n^r)$.  In addition,  let $\mathcal{E}$ be defined as in previous section. Notice that there could be two $\mathcal{H}$-subgraphs $H_1$, $H_2$ in $G$ with the same element $V \in \mathcal{V}$ but different set of edges $E_1,E_2 \in \mathcal{E}$. In this way, $V(E_1)=V(E_2)$ but $E_1 \neq E_2$.

%In addition,  $\mathcal{E}$ defines a collection of sets of edges each of which induces a subgraph isomorphic to some $H \in \mathcal{H}$. That is, for each $E \in \mathcal{E}$, there is an $H \subseteq G[V(E)]$ that is an $\mathcal{H}$-subgraph. In this way, for each $E \in\mathcal{E}$, $|E| \leq r^2$, and $|\mathcal{E}| \leq n^{r^2}$. 

%Slightly abusing our  notation, let $V(E)$ denote the set of vertices corresponding to the end-points of the edges $E$.  Notice that there could be two $\mathcal{H}$-subgraphs $H_1$, $H_2$ in $G$ with the same element $V \in \mathcal{V}$ but different set of edges $E_1,E_2 \in \mathcal{E}$. In this way, $V(E_1)=V(E_2)$ but $E_1 \neq E_2$.

%\marginpar{\textcolor{blue}{I think for this version of the problem, $H \in \mathcal{H}$ must be connected. For example if $H$ is a $K_3$ + an isolated vertex. The sets in $\mathcal{S}$ would contain the vertices for the $K_3$ and the isolated vertex would not be included in the set. }}

\begin{transf}\label{AltransfForEdgeMembership}
\textbf{Input:} $G$, $\mathcal{H}$;
\textbf{Output:} $\mathcal{U}$, $\mathcal{S}$, and $r$.

The universe $\mathcal{U}$ equals $(E(G)\times\{1,\dots,t\})\cup \mathcal{V}$. 

Let $r= m(\mathcal{H})$.

The collection $\mathcal{S}$ contains all subsets of $\mathcal{U}$ each with at most $r+1$ elements $\{\{(e_1,j_1), \dots, (e_{m'},j_{m'}), V(E)\}$ $\mid E \in \mathcal{E}$, $E=\{e_1,\dots e_{m'}\}$ for each $1 \leq j_i \leq t$ and $1 \leq i \leq m'$, where $m' \leq m(\mathcal{H})\}$.
\end{transf}

The intuition behind our transformation is that each set in $\mathcal{S}$ would represent an $\mathcal{H}$-subgraph in $G$. Since each edge in $E(G)$ can be in at most $t$ $\mathcal{H}$-subgraphs, we replicate  each edge $t$ times. This would imply that there are $t^{r}$ sets in $\mathcal{S}$ representing the same $\mathcal{H}$-subgraph $H$ of $G$. All of those sets have a common element which is precisely the set of vertices in $V(H)$, i.e., some $V \in \mathcal{V} \subset \mathcal {U}$. This will avoid to choose sets that corresponds to $\mathcal{H}$-subgraphs with identical set of vertices.

%Given that we are adding an element $V(E)$ to the sets of $\mathcal{S}$, a pair of $\mathcal{H}$-subgraphs $H_1,H_2$ with $V(H_1)=V(H_2)$ will correspond to a pair of sets that pairwise intersect in $V(H_1)=V(H_2)$. In this way, both cannot be contained in a a $k$-Set Packing and therefore they will not be together in a $(k,t)$-Edge-$\mathcal{H}$-Membership.

\begin{example}
For the graph of Figure \ref{membership2}, $\mathcal{H}=\{C_4\}$, $k=3$, and $t=2$.
\begin{eqnarray*}\mathcal{V}&=&\left\{\{a,b,c,d\}, \{b,c,e,d,f\}, \{b,c,g,h\}\right\}\\
\mathcal{U} &=&\{{(ab,1), (ad,1), (dc,1), (bc,1), (be,1), (ef,1), (cf,1), (bf,1), (ec,1), (gb,1), }\\
             &&\ {(dc,1), (hc,1), (ab,2), (ad,2), (dc,2), (bc,2), (be,2), (ef,2), (cf,2), (bf,2), }\\
             &&\ {(ec,2), (gb,2), (dc,2),(hc,2)} \}\\
&\cup&\mathcal{V}\,. 
\end{eqnarray*}

%\begin{example}
%For the instance of Figure \ref{membership2},
%\begin{eqnarray*}\mathcal{V}&=&\left\{\{a,b,c,d\}, \{b,c,e,d,f\}, \{b,c,g,h\}\right\}\\
%\mathcal{U} &=&\{{\tiny (ab,1), (ad,1), (dc,1), (bc,1), (be,1), (ef,1), (cf,1), (bf,1), (ec,1), (gb,1), (dc,1), (hc,1),}\\
%&&\ {\small (ab,2), (ad,2), (dc,2), (bc,2), (be,2), (ef,2), (cf,2), (bf,2),  (ec,2), (gb,2), (dc,2),(hc,2)}\}\\
%&\cup&\mathcal{V}\,. 
%\end{eqnarray*}

%$\mathcal{V}=$$\{$$\{a,b,c,d\}$, $\{b,c,e,d,f\}$, $\{b,c,g,h\}$$\}$ and $\mathcal{U} =\{(ab,1)$, $(ad,1)$, $(dc,1)$, $(bc,1)$, $(be,1)$, $(ef,1)$, $(cf,1)$, $(bf,1)$, $(ec,1)$, $(gb,1),$ $(dc,1)$, $(hc,1)$,
%$(ab,2)$, $(ad,2)$, $(dc,2)$, $(bc,2)$, $(be,2)$, $(ef,2)$, $(cf,2)$, $(bf,2),$  $(ec,2)$, $(gb,2)$, $(dc,2)$, $(hc,2)$  $\cup$ $\mathcal{V}\}$. 

Some sets of the collection $\mathcal{S}$ are 

$\{(ad,1),$$(ab,1),$$(bc,1),$$(dc,1),$$\{a,b,c,d\}\}$, 

$\{(bc,1),$$(cf,1),$$(ef,1),$$(be,1),$$\{b,c,e,f\}\}$, 

$\{(bf,1),$$(cf,1),$$(ec,1),$$(be,1),$$\{b,c,e,f\}\}$,

$\{(bc,1),$$(bf,1),$$(ef,1),$$(ec,1),$$\{b,c,e,f\}\}$,

$\{(gh,1),$$(dc,1),$$(bc,1),$$(hc,1),$$\{b,c,g,h\}\}$.

A $(3,4+1,0)$-set packing is $\{$$\{(ab,1)$, $(ad,1)$, $(dc,1)$, $(bc,1)$, $\{a,b,c,d\}\}$, $\{(be,1)$, $(bf,1)$, $(cf,1)$, $(ec,1)$, $\{b,c,e,f\}\}$ and $\{(bc,2)$, $(gb,1)$, $(hc,1)$, $(gh,1)$, $\{b,c,g,h\}\}$$\}$. 

This corresponds to the following $(3,4,2)$-edge-$\mathcal{H}$-membership \\ $\{H_1=(\{a,b,c,d\}$, $\{ab,ad,dc,bc\})$, $H_2=(\{b,c,e,f\}$, $\{be,bf,cf,ec\})$ and  $H_3=(\{b,c,g,h\}$, $\{bc,gb,gh,hc\})$. 
\end{example}

\begin{lemma}
$G$ has a $(k,r,t)$-edge-$\mathcal{H}$-membership (NISV) if and only if $\mathcal{S}$ has a $(k,r+1,0)$-set packing.
\end{lemma}

\begin{proof}
We build a $(k,r+1,0)$-set packing $\mathcal{K_S}$ from a $(k,r,t)$-edge-$\mathcal{H}$-membership (NISV) $\mathcal{K}$. For each $H_j \in \mathcal{K}$, we add a set $S_j$ with $|E(H_j)|+1 \leq r+1$ elements as follows.  The element $V$ in $S_j$ corresponds to the set of vertices $V(H_j)$. Note that $V(H_j) \in \mathcal{V} \subseteq \mathcal{U}$. For each edge $e_i \in E(H_j)$ ($1 \leq i \leq |E(H_j)| \leq m$), we add an ordered pair $(e_i,l+1)$ to $S_j$ where $l$ corresponds to the number of $\mathcal{H}$-subgraphs in $\{H_1, \dots, H_{j-1}\} \subset \mathcal{K}$ that $e_i$ belongs to. Since each edge is a member of at most $t$ $\mathcal{H}$-subgraphs, $0 \leq l \leq t-1$ and $(e_i,l+1)$ always exists in $\mathcal{U}$. 
By our construction of $\mathcal{S}$,  $S_i \in \mathcal{K_S} \subseteq \mathcal{S}$. 
It remains to show that the sets in $\mathcal{K_S}$ are pairwise disjoint. Since $V(H_i) \neq V(H_j)$ for each pair $H_i,H_j \in \mathcal{K}$, none of the sets in $\mathcal{K_S}$ will share the same element $V$. On the other hand, no pair $S_i,S_j$ shares an ordered pair $(e,l)$ for some edge $e \in E(G)$ and some integer $1 \leq l \leq t$. If $S_i$ and $S_j$ share a pair $(e,l)$, this would imply that there are two $\mathcal{H}$-subgraphs $H_i$ and $H_j$ that share the same edge $e$.  Without lost of generality, assume that $H_i$ appears before  $H_j$ in $\mathcal{K}$. If $(e,l)$ is contained in both $S_i$ and $S_j$,   then $e$ is a member of $l-1$ $\mathcal{H}$-subgraphs from both $\{\dots, H_i, \dots \}$ and $\{\dots, H_i, \dots, H_j \}$. This is a  contradiction since both $H_i$ and $H_j$ contain $e$.

We construct a $(k,r,t)$-edge-$\mathcal{H}$-membership (NISV) $\mathcal{K}$ using a $(k,r+1,0)$-set packing $\mathcal{K_S}$. Each set $S_i \in \mathcal{K}_S$ has an element $V_{S_i} \in \mathcal{V}$ which corresponds to a set of vertices in $V(G)$ that forms a subgraph  in $G$ that is isomorphic to some $H\in \mathcal{H}$.  In addition, we take the edge $e_l$ that appears in each $(e_l,j) \in S_i$ (for $1 \leq l \leq m$ and $j \in [1...t]$) and we add it to a set $E_{S_i}$.
%In addition, we denote by $E_{S_i}$ the union of the edges\marginpar{HF: What does "union of the edges" formally mean?} of the ordered pairs in $S_i$, i.e., $E_{S_i} = \bigcup (e_i,j)$ for each $(e_i,j) \in S_i$ $(1 \leq i \leq m)$ and $j \in [1 \dots t]$. 
In this way, for each set $S_i \in \mathcal{K}_S$, we add to $\mathcal{K}$ an $\mathcal{H}$-subgraph $H_{S_i}$ with a set of vertices $V_{S_i}$ and a set of edges $E_{S_i}$, i.e., $H_{S_i}=(V_{S_i},E_{S_i})$.  Since the sets in $\mathcal{K}_S$ are pairwise disjoint, $V(H_{S_i}) \neq V(H_{S_j})$.  Now, we need to show that each edge of $E(G)$ is a member of at most $t$ $\mathcal{H}$-subgraphs in $\mathcal{K}$. Assume the contrary for the sake of contradiction. If two $\mathcal{H}$-subgraphs in $\mathcal{K}$ share an edge $e$, then there are two sets in $\mathcal{K}_{S}$ where each one has a pair $(e,i)$ and $(e,j)$, respectively. Since the sets in $\mathcal{K}_S$ are disjoint, $i \neq j$. If an edge $e \in E(G)$ is a member of more than $t$ $\mathcal{H}$-subgraphs in $\mathcal{K}$, then there are at least $t+1$ ordered pairs $(e,1),\dots,(e,t+1)$ contained in $t+1$ different sets in $\mathcal{K}_S$. However, by our construction of $\mathcal{U}$, there are at most $t$ ordered pairs that contain $e$. \qed 
\end{proof}

The size of our constructed instance is upper-bounded by $$|\mathcal{U}|=|E(G)| \times t + |\mathcal{V}|\leq
t n^2 + n^{r(\mathcal{H})} = O(n^{r(\mathcal{H})})\,.$$ Each set in $\mathcal{S}$ has size $r+1$ (where $r=m(\mathcal{H})$) and $|\mathcal{S}| \leq t^{r} |\mathcal{H}| n^{m(\mathcal{H})}$. 

We construct a reduced graph $G'$ from a reduced universe $\mathcal{U'} \subseteq \mathcal{U}$ as follows.

\begin{transf}\label{FromUniverseToGraph_3}
\textbf{Input:} $G$, $\mathcal{U'}$;
\textbf{Output:} $G'$

We take the end-points of each edge $e$ that appears in each $(e,i) \in\mathcal{U}'$ and the vertices in each $V \in \mathcal{V}'$ and consider the graph $G'$ that is induced by all these in $G$.
\end{transf}

\begin{lemma}
$G'$ has a $(k,r,t)$-$\mathcal{H}$-membership (NISV) if and only if $\mathcal{S'}$ has a $(k,r+1,0)$-set packing.
\end{lemma}

The kernelization algorithm for the $\mathcal{H}$-Packing with $t$-Edge Membership problem (NISV) corresponds to Algorithm \ref{GenericKernelization} with input: $G$, $\mathcal{H}$, $k$, ``$(r+1)$-Set Packing'', Transformation \ref{AltransfForEdgeMembership}, kernelization algorithm \cite{Faisal10}, and Transformation \ref{FromUniverseToGraph_3}. The size of $\mathcal{U'}$ is at most $2r!((r+1)k-1)^{r}$ \cite{Faisal10}; hence $|V(G')|=O((r+1)^{r} k^{r})$. Therefore, we can state:

%our kernelization reduces our version of the $\mathcal{H}$-Packing with $t$-Edge Membership that does not allow $\mathcal{H}$-subgraphs with identical set of vertices in a $(k,t)$-Edge-$\mathcal{H}$-Membership to a kernel with $O(r^{r} k^{r-1})$ vertices.

%which matches the result of Theorem \ref{KernelEdgeMembership}.

%Similarly as in Subsection \ref{membershipSets}, we run the kernelization algorithm for the Set Packing problem  \cite{Faisal10}. This algorithm would leave us with a new universe $\mathcal{U}'\subseteq \mathcal{U}$ with at most $2(r-1)!(rk-1)^{r-1}$ elements, as well with a collection $\mathcal{S}'$ of subsets.
%The following property is borrowed from \cite{Faisal10}. 

%\begin{lemma}\label{kernelMembershipEdge}
% $\mathcal{S}$ has a $k$-Set Packing if and only if $\mathcal{S'}$ has a $k$-Set Packing.
%\end{lemma}

%Next we construct a  graph $G'$ using $\mathcal{U'}$. To this end, we take the end-points of each edge $e$ that appears in each $(e,i) \in\mathcal{U}'$ and the vertices in each $V \in \mathcal{V}'$ and consider the graph $G'$ that is induced by all these in $G$. By our construction, $G'$ has at most $2(r-1)!(rk-1)^{r-1}$ vertices. 
%The graph $G'$, together with $k$, is our desired kernel, a $\mathcal{H}$-Packing with $t$-Membership instance.

%\begin{lemma}
%$G'$ has a $(k,t)$-Edge-$\mathcal{H}$-Membership if and only if $\mathcal{S'}$ has a $k$-Set Packing.
%\end{lemma}

%This allows us to conclude:

\begin{theorem}
The $\mathcal{H}$-Packing with $t$-Edge Membership problem (NISV) has a kernel with $O((r+1)^{r} k^{r})$ vertices, where $r=m(\mathcal{H})$.
\end{theorem}

%% file: Pairwise_SetPacking_Kernel.tex
%\vspace{0.25cm}

%\fbox{
%\parbox{11cm}{
%\textbf{The $r$-Set Packing with $t$-Overlap problem}
    
%    \noindent \emph{Instance}: A collection $\mathcal{S}$ of distinct sets, each of size at most $r$, drawn from a universe $\mathcal{U}$ of size $n$, a positive integer $k$.

%    \noindent \emph{Parameter}: $k$.

%    \noindent \emph{Question}: Does $\mathcal{S}$ contain a $(k,r,t)$-set packing, i.e., a collection of at least $k$ sets $\mathcal{K}=\{S_1,\dots , S_k\}$ where $|S_i \cap S_j| \leq t$, for any pair $S_i,S_j$ with $i\neq j$?
%}
%}
%\vspace{0.25cm}

%Observe that for $t=0$, we are back to the classical $r$-Set Packing problem.

We assume that each set in $\mathcal{S}$ has size at least $t+1$. Otherwise, we can add the sets with size at most $t$ straight to a $(k,r,t)$-set packing and decrease the parameter $k$ by the number of those sets. We start with a simple reduction rule. 

\begin{redrule}\label{cleanupSetPacking}
Remove any element of $\mathcal{U}$ that is not contained in at least one set of $\mathcal{S}$.
\end{redrule}

Notice that every pair of sets in $\mathcal{S}$ overlap in at most $r-1$ elements (even if $S' \subset S$ for some $S',S \in \mathcal{S}$); otherwise there would be two identical sets. Thus for $t=r-1$, if $|\mathcal{S}|\geq k$ then $\mathcal{S}$ is a $(k,r,t)$-set packing; in the other case, $\mathcal{S}$ does not have a solution. Henceforth, we assume that $t \leq r-2$. Algorithm \ref{iterative} reduces an instance of the Set Packing with $t$-Overlap problem to a kernel in two steps. First (Lines 2-9), the goal is to identify sets of $\mathcal{S}$ that \emph{may not be needed} (\emph{extra}) to form a $(k,r,t)$-set packing.  In Line 16, \emph{unnecessary} elements are removed  from $\mathcal{U}$
by triggering a reduction.

We next explain the idea in Lines 2-9. First, we compute a maximal $(r,r-2)$-set packing $\mathcal{R}$ of $\mathcal{S}$, i.e., 
a maximal collection of sets from $\mathcal{S}$ such that every pair of sets in $\mathcal{R}$ overlaps in at most $r-2$ elements. We say that a set $S_e \in \mathcal{S}$ is \emph{extra} if there is a $(k,r,t)$-set packing that does not include $S_e$. Algorithm \ref{extrasubgraphsred} (Line 6) identifies \emph{extra} sets in $\mathcal{R} \subseteq \mathcal{S}$. We could just use $\mathcal{R}=\mathcal{S}$ as input of Algorithm \ref{extrasubgraphsred}, (i.e., a maximal-$(r,r-1)$-set packing instead). However, a smaller set $\mathcal{R}$ will lead to a smaller problem kernel. 

Lines 3-9 in Algorithm \ref{iterative} basically keep reducing the set $\mathcal{R}$ using Algorithm \ref{extrasubgraphsred} while there are no more sets from $\mathcal{S} \backslash (\mathcal{R} \cup \mathcal{E}$) to add to $\mathcal{R}$. We need to run  Algorithm~\ref{extrasubgraphsred} possibly more than once in order to preserve the maximality of $\mathcal{R}$. That is, after Line 4, each set in $\mathcal{S} \backslash (\mathcal{R} \cup \mathcal{E}$)  overlaps in $r-1$ elements with at least one set in $\mathcal{R}$. After Line 7, $\mathcal{R}$ is reduced by $\mathcal{E'}$, and $\mathcal{E}$ is updated with $\mathcal{E'}$. Therefore, all sets in $\mathcal{S} \backslash (\mathcal{R} \cup \mathcal{E})$ that were overlapping in $r-1$ elements only with sets in $\mathcal{E'}$ can now be added to $\mathcal{R}$. 

%%Example******************************************************
\begin{example}\label{instanceExample}
Consider the following instance of the $r$-Set Packing with $t$-Overlap problem, where $r=4$, $t=2$, and $k=2$. \\

$\mathcal{U}=\{a,b,c,d,e,f,g,h,i,j,k,l,m,n,o,p,q,r,s,t,u,v,w,x,y\}$ and

$\mathcal{S}=\{\{a,b,c,e\},\{b,c,d,e\},\{e,f,g,i\},\{a,e,f,i\},\{e,g,i,h\},\{i,j,n,m\},$

$\{i,j,m,k\},\{i,j,m,l\},\{o,q,n,p\},\{q,p,r,s\},\{q,p,t,u\},\{q,p,u,v\},\{q,p,v,w\},$

$\{q,p,x,y\}\}$. \\

Figure \ref{hypergraph} shows the hypergraph constructed from $\mathcal{U}$ and $\mathcal{S}$. 

A maximal $(r,r-2)$-set packing is $\mathcal{R} =\{ \{b,c,d,e\},\{e,f,g,i\},\{i,j,n,m\},$
$\{o,n,p,q\},\{q,p,r,s\},\{q,p,t,u\},\\ \{q,p,v,w\},\{q,p,x,y\} \}$. This maximal set is highlighted with bold lines in Figure \ref{hypergraph}. 
\end{example}
%End of Example****************************************************

\begin{figure}[htb]     
     \centerline{{\includegraphics[scale=0.60]{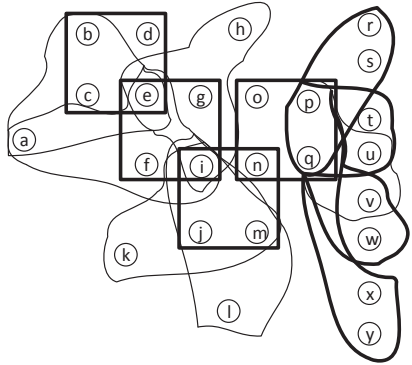}}}
     \caption{The hypergraph constructed with $\mathcal{S}$ and $\mathcal{U}$ from Example \ref{instanceExample}} \label{hypergraph}
\end{figure}

%***********************************************************************************
%Kernelization algorithm------------------------------------------------------------

\begin{algorithm} 
  \caption{Kernelization Algorithm - Set Packing with $t$-Overlap} 
	\textbf{Input:} An instance $\mathcal{U},\mathcal{S}$ \\
	\textbf{Output:} A reduced instance $\mathcal{U'},\mathcal{S'}$
	\begin{algorithmic}[1]
	  
		\STATE{Apply Reduction Rule \ref{cleanupSetPacking}}
		
		\STATE{$\mathcal{R} = \emptyset$, $\mathcal{E} = \emptyset$}
		
		\REPEAT
			
		\STATE{Greedily add sets from $\mathcal{S} \backslash (\mathcal{R} \cup \mathcal{E})$ to $\mathcal{R}$ such that every pair of sets in $\mathcal{R}$ overlaps in at most $r-2$ elements (i.e., a maximal $(r,r-2)$-set packing).}
		
		\IF{at least one set was added to $\mathcal{R}$}
		
		    \STATE{$\mathcal{E'}$ = Algorithm \ref{extrasubgraphsred}$(\mathcal{R})$}
		
		    \STATE{$\mathcal{R} = \mathcal{R} \backslash \mathcal{E'}$, $\mathcal{E} = \mathcal{E} \cup \mathcal{E'}$}
				
		\ENDIF
		
		\UNTIL{no more sets have been added to $\mathcal{R}$}
		
		\STATE{Reduce $\mathcal{S} = \mathcal{S} \backslash \mathcal{E}$ and re-apply Reduction Rule \ref{cleanupSetPacking}}
		
		\STATE{Compute a maximal $(r,t)$-set packing $\mathcal{M}$ of $\mathcal{R}$} 
		
		 \IF{$|\mathcal{M}| \geq k$}
		
		    \STATE{Let $\mathcal{S'}$ be any subset of $\mathcal{M}$ of size $k$}
				
		    \STATE{Return $\mathcal{U}$ and $\mathcal{S'}$}
		
		 \ENDIF 
						
		\STATE{Let $\mathcal{U'}$ and $\mathcal{S'}$ be the reduced universe and collection of sets, respectively, after applying Reduction Rule \ref{matching}.}
		
		\STATE{Return $\mathcal{U'}$ and $\mathcal{S'}$}
		
		\end{algorithmic} \label{iterative}
\end{algorithm} 

The basic idea of Algorithm \ref{extrasubgraphsred} is that if there is more than a specific number of members of $\mathcal{S}$ that pairwise overlap in exactly the same subset of elements $P \subseteq \mathcal{U}$, then some of those members are extra. This approach has been applied to obtain kernels for vertex-disjoint packing problems \cite{Fellows04b,Moser09}. To our knowledge, this is first generalization for packing problems with arbitrary overlap.

For each set $S$ in $\mathcal{R}$, Algorithm \ref{extrasubgraphsred} takes every subset of elements $P \subsetneq S$ from sizes from $t_{Ini}$ to $t+1$ in decreasing order (Lines 3-5). The variable $t_{Ini}$ is defined as the maximum overlap value between any pair of sets in $\mathcal{R}$ (Line 1). Thus, given the construction of $\mathcal{R}$ in Algorithm \ref{iterative}, $t_{Ini}=r-2$. Note that $|S|>i$ to run Lines 4-13; otherwise $S$ will be considered at a later iteration of the for-loop of Line 3. For every $P$, we count the number of sets in $\mathcal{R}$ that contain $P$ (collected in  $\mathcal{P}$, Line 6). If that number is greater than a specific threshold (Line 7) then we can remove the ones above the threshold (Lines 8-11). The size of $P$ is determined by the variable $i$ in Line 3, and the threshold $f(i)$ is defined as $f(i) = (r-t)(k-1) f(i+1) + 1$ where $f(t_{Ini}+1)$ is initialized to one (Line 2).  The function $f(i)$ basically bounds the number of sets in $\mathcal{R}$ that contains a subset $P$ with $i$ elements. Observe that $f(i) = (r-t)(k-1) f(i+1) + 1$ $= \sum_{j=0}^{t_{Ini}-i+1} [(r-t)(k-1)]^j$. Algorithm \ref{extrasubgraphsred} returns $\mathcal{E}$ which is the set of extra sets in $\mathcal{R}$ (Line 16).

\begin{algorithm} 
  \caption{Extra Sets Reduction} 
	\textbf{Input:} A set $\mathcal{R} \subseteq \mathcal{S}$ of sets \\
	\textbf{Output:} A set $\mathcal{E} \subseteq \mathcal{R}$  of sets
	\begin{algorithmic}[1] 
		
		\STATE{$t_{Ini}= \max{\{|S_i\cap S_j| : S_i,S_j \in \mathcal{R}, i\neq j\} }$}
		
		\STATE{$f(t_{Ini}+1) = 1$, $\mathcal{E}=\emptyset$} 
				
		\FOR{$i=t_{Ini}$ downto $t+1$}     %%pattern size
			
			\FOR{each $S \in \mathcal{R}$ such that $|S| > i$}   %%each H-subgraph 
			
					\FOR{each $P \subsetneq S$ where $|P|=i$}   %%each pattern
					
						\STATE{$\mathcal{P} =\{S' \in \mathcal{R} : S' \supsetneq P\}$} 
						
						 \STATE{$f(i) = (r-t)(k-1) f(i+1) + 1$}
						
						 \IF{$|\mathcal{P}|$ $> f(i)$}
						
						     \STATE{Choose any $\mathcal{P'} \subset \mathcal{P}$ of size $f(i)$}
								
								 \STATE{Set $\mathcal{E} \leftarrow \mathcal{E} \cup (\mathcal{P}\backslash \mathcal{P'})$ (extra sets)}
						     
								 \STATE{$\mathcal{R} \leftarrow \mathcal{R} \backslash (\mathcal{P}\backslash \mathcal{P'})$}
								
						\ENDIF
				 \ENDFOR			
		   \ENDFOR
		\ENDFOR
	\STATE{Return $\mathcal{E}$}

	 \end{algorithmic} \label{extrasubgraphsred}
\end{algorithm} 

Next we show that Algorithm \ref{extrasubgraphsred} correctly reduces a given maximal set $\mathcal{R}$.

%===================================================================================
%Correctness Lemmas
%==================================================================================

%**********************Correctness of Algorithm 2 (redundant subgraphs reduction)

\begin{lemma}\label{solutionsets}
$\mathcal{R}$ has a $(k,r,t)$-set packing if and only if $\mathcal{R} \backslash \mathcal{E}$ has a $(k,r,t)$-set packing.
\end{lemma}

\begin{proof}
%Intro---------------------------------------------------------------------------------------------------------
Since $\mathcal{R}\backslash\mathcal{E} \subseteq \mathcal{R}$, we know that if $\mathcal{R}\backslash\mathcal{E}$ has a $(k,r,t)$-set packing then $\mathcal{R}$ has one. Let $\mathcal{K}$ be a $(k,r,t)$-set packing of $\mathcal{R}$. If $\mathcal{K}$ is also a $(k,r,t)$-set packing of $\mathcal{R}\backslash\mathcal{E}$ then the lemma follows. Otherwise, there is at least one set $S_e \in \mathcal{K}$ that is extra, ($S_e \in \mathcal{E}$ by Line 10); that is, $\mathcal{K} \cap \mathcal{E}\neq\emptyset$. We claim that we can replace each set in $\mathcal{K} \cap \mathcal{E}$ (an extra set) by a set in $\mathcal{R}\backslash\mathcal{E}$. 

%Why we choose i---------------------------------------------------------------------------------------------------
For any set $S_e$ removed from $\mathcal{R}$ (i.e., $S_3 \in \mathcal{E}$), there is at least one subset of elements $P \subset S_e$ such that $\mathcal{P} > f(i)$ (Line 8). Since the variable $i$ determines the size of $P$ and $f(i)$ is an upper bound for the number of sets containing $P$, we will prove the lemma by induction on $i$.

%First iteration.---------------------------------------------------------------------------------------------------
Let $i=t_{Ini}$ (first iteration of loop of Line 3) and let $P$ be a subset with $t_{Ini}$ elements such that $\mathcal{P} > f(i)$, where $\mathcal{P}$ is the set of all sets in $\mathcal{R}$ that contain $P$ (Line 6).

Given that $|P|=t_{Ini} \geq t+1$ (Line 3), at most one set in $\mathcal{P}$ can be in a $(k,r,t)$-set packing $\mathcal{K}$. Otherwise, there would be a pair of sets in $\mathcal{K}$ overlapping in more than $t$ elements. Assume that set is $S_e$ and is in $\mathcal{E}$ (Lines 9-10).

We claim that we can replace $S_e$ in $\mathcal{K} \cap \mathcal{E}$ by a set in $\mathcal{P'}$ (the members of $\mathcal{S}$ that contain $P$ that were kept in $\mathcal{R}$, Line 11).

\begin{claim}
The $k-1$ sets in $\mathcal{K} \backslash S_e$ conflict with at most $(r-t)(k-1)$ sets in $\mathcal{P'}$.
\end{claim}

\begin{claimproof}
A set $S^* \in \mathcal{K} \backslash S_e$ conflicts with a set $S' \in \mathcal{P'}$ if they overlap in more than $t$ elements, i.e., $|S^* \cap S'| \geq t+1$.  

Given that $P \subset S_e$ and $|S_e \cap S^*| \leq t$, $|P \cap S^*| = l \leq t$. Therefore, at least $t+1 - l$ elements from $S^*$ should be in $S' \backslash P$ to have a conflict between $S^*$ and $S'$. That is, $| (S'\backslash P) \cap S^*| \geq t+1 -l$. Let us denote as $u_1 \dots u_{t+1 - l}$ the elements shared between $S'\backslash P$ and $S^*$. 

Every set in $\mathcal{P'}$ that contain the subset of elements $P \cup u_1 \dots u_{t+1 - l}$ will conflict with $S^*$. The sets in $\mathcal{P'}$ pairwise overlap exactly in the set $P$. Otherwise, there would be a pair of sets in $\mathcal{P} \subset \mathcal{R}$ overlapping in more than $t_{Ini}$ elements ($|P|=t_{Ini}$), a contradiction by Line 1.  This implies that there is at most one set in $\mathcal{P'}$ that contain the subset of elements $P \cup u_1 \dots u_{t+1 - l}$. This proves our initialization of $f(t_{Ini})=1$.

It remains to count the number of disjoint subsets of $S^*$ with $t+1-l$ elements (That is, all possible disjoint $u_1 \dots u_{t+1 - l}$). Notice that only the disjoint sets are relevant for counting as we showed that the sets in $\mathcal{P'}$ intersect in exactly $P$.

From the $|S^*|-l \leq r-l$ elements ($l$ are already in $P$) we can have at most $\frac{r-l}{t+1-l}$ disjoint sets of $t+1-l$ elements. Hence,  there are at most $\frac{r-l}{t+1-l}$ sets in $\mathcal{P'}$ that could conflict with $S^*$. Repeating this argument, the $k-1$ sets in $\mathcal{K} \backslash S_e$ conflict with at most $\frac{r-l}{t+1-l} (k-1)$ sets in $\mathcal{P'}$.  This expression is maximum and equal to $(r-t)(k-1)$ when $l=t$. \hfil $\Diamond$
\end{claimproof}

Given that $|\mathcal{P'}|=(r-t)(k-1) f(t_{Ini}+1) +1$ (Line 7) (where $f(t_{Ini}+1)=1$, Line 2), there is at least one set in $\mathcal{P'}$ that will not conflict with any set in $\mathcal{K} \backslash S_e$ and it can replace $S_e$.

%The rest of the iterations---------------------------------------------------------------------------------------------------
Now, we will develop the same argument for $i<t_{Ini}$. 

Let $i < t_{Ini}$ and $P$ be a subset of elements of size $i$ such that $\mathcal{P} > f(i)$. Since the size of $|P| \geq t+1$ and all sets in $\mathcal{P}$ share at least $P$ (Line 6), at most one set of $\mathcal{P}$ is in a $(k,r,t)$-set packing. Assume that set is $S_e$ and is extra, i.e.,  $S_e \in \mathcal{P} \backslash \mathcal{P'}$.

\begin{claim}
The $k-1$ sets in $\mathcal{K}\backslash S_e$ conflict with at most $(r-t)(k-1)f(i+1)$ sets in $\mathcal{P'}$.
\end{claim}

\begin{claimproof}
A set $S^* \in \mathcal{K} \backslash S_e$ conflicts with a set $S' \in \mathcal{P'}$ if $|S^* \cap S'| \geq t+1$. Given that $P \subset S_e$ and $|S_e \cap S^*| \leq t$, $|P \cap S^*| = l \leq t$. Therefore, at least $t+1 - l$ elements from $S^*$ should be in $S' \backslash P$ to have a conflict between $S^*$ and $S'$. That is, $|(S'\backslash P) \cap S^*| \geq t+1 -l$. Let us denote as $u_1 \dots u_{t+1 - l}$ the set of elements shared between $S'\backslash P$ and $S^*$. 

Notice that all sets in $\mathcal{P'}$ that contain the subset of elements $P \cup u_1 \dots u_{t+1 - l}$ will conflict with $S^*$. There are $f(i + (t+1-l))$ sets in $\mathcal{R}$ that contain a subset of elements of size $i + (t+1-l)$. Therefore, there are at most $f(i + (t+1-l))$ sets in $\mathcal{P'}$ that contain the subset $P \cup u_1 \dots u_{t+1 - l}$ and will conflict with $S^*$.

It remains to compute the number of subsets with $t+1-l$ elements from $S^*$. There are at most $\binom{r-l}{(t+1)-l}$ ways of selecting $t+1-l$ elements from the $|S^*-l|\leq r-l$ elements of $S^*$. Repeating this argument, we find that the $k-1$ sets in $\mathcal{K}\backslash S_e$ can conflict with at most 

\begin{equation}\label{boundgeneral}
\binom{r-l}{t+1-l}(k-1)f(i+(t+1-l)) 
\end{equation}

\noindent 
sets of $\mathcal{P'}$.

The function $f(i+j)$ is equivalent to $\sum_{j=0}^{t_{Ini}-(i+j)+1} [(r-t)(k-1)]^j$ which is upper-bounded by $2[(r-t)(k-1)]^{t_{Ini}-i-j+1}$. This implies that $f(i) > f(i+1) > f(i+2) > \dots f(i+j)$. Therefore Expression \ref{boundgeneral} is maximum when $l=t$ and upper-bounded by $(r-t)(k-1) f(i+1)$. \hfill $\Diamond$
\end{claimproof}

Since there are $(r-t)(k-1)f(i+1)+1$ sets in $\mathcal{P'}$, there is at least one set in $\mathcal{P'}$ that would not conflict with any set in $\mathcal{K}\backslash S_e$ and it can replace $S_e$. \qed

\end{proof}

%%Example******************************************************
\begin{example}
Consider iteration $i=r-2$, $S=\{o,n,p,q\}$, and $P=\{p,q\}$ in Algorithm \ref{extrasubgraphsred} for the instance of Example \ref{instanceExample}. In this case,  $$\mathcal{P}=\{\{o,n,p,q\},\{q,p,r,s\},\{q,p,t,u\}, \{q,p,v,w\},\{q,p,x,y\}\}\,.$$ Since $k=2$, for $i=r-2$, $f(r-2)=3$. Thus, one set of $\mathcal{P}$ will be removed in Line 9.  Without lost of generality, assume this set is $\{p,q,x,y\}$. For this example, that is the only set that will be removed from $\mathcal{R}$.
\end{example}
%, and therefore from $\mathcal{S}$ in Line 11 of Algorithm \ref{iterative}. 
%Figure \ref{} shows the updated hypergraph after this removal.
%%End of Example******************************************************

%*****************Iterative execution of Algorithm 2 (redundant subgraphs reduction)
We prove next that the set $\mathcal{E}$ found in Lines 3-9 of Algorithm \ref{iterative} is extra.

\begin{lemma}\label{iterativeLemma}
After running Lines 1-9 of Algorithm \ref{iterative}, $\mathcal{S}$ has a $(k,r,t)$-set packing if and only if $\mathcal{S} \backslash \mathcal{E}$ has a  $(k,r,t)$-set packing.
\end{lemma}

\begin{proof}
Since $(\mathcal{S} \backslash \mathcal{E}) \subseteq \mathcal{S}$, if $\mathcal{S} \backslash \mathcal{E}$ has a $(k,r,t)$-set packing then $\mathcal{S}$ has one. Now, assume by contradiction that $\mathcal{S}$ has a $(k,r,t)$-set packing $\mathcal{K}$ but $\mathcal{S} \backslash \mathcal{E}$ does not have one. This implies that there is an $S \in \mathcal{K}$ that is in $\mathcal{E}$.
%\marginpar{HF: Please specify} 
If that set $S$ is in $\mathcal{E}$ then in some iteration of Algorithm \ref{iterative} (Lines 3-9) that $S$ was in $\mathcal{R}$ in Line 5 and after that $S$ was not in $\mathcal{R}$ but in $\mathcal{E'}$. However, by Lemma \ref{solutionsets} we know that $\mathcal{R}$ has a $(k,r,t)$-set packing if and only if $\mathcal{R} \backslash \mathcal{E'}$ has one.  \qed
\end{proof}

%********************Independent Set Correctness
By Lemma \ref{iterativeLemma}, we now reduce $\mathcal{S}$ by removing $\mathcal{E}$ in Line 10 of Algorithm \ref{iterative}. As a consequence, we re-apply Reduction Rule \ref{cleanupSetPacking}. Then, we compute a maximal $(r,t)$-set packing in $\mathcal{R}$. As we will see in proof of Lemma \ref{sizeHgR}, this maximal solution will help us to determine an upper-bound for the number of sets in $\mathcal{R}$. From now on, we assume that the maximal solution was not a $(k,r,t)$-set packing and Algorithm \ref{iterative} continues executing.

An element $u \in \mathcal{U}$ is \emph{extra}, if there is a $(k,r,t)$-set packing of $\mathcal{S}$ where $u \notin S$ for each set $S \in \mathcal{K}$. We will identify extra elements in $\mathcal{U} \backslash val(\mathcal{R})$, denoted as $O$ henceforth. Before giving our reduction rule, we give a characterization of $O$.

\begin{lemma}\label{contained}
Let $O=\mathcal{U} \backslash val(\mathcal{R})$. 
(i) Each element in $O$ is contained in at least one set $S$.
(ii) Only sets in $\mathcal{S} \backslash \mathcal{R}$ contains elements from $O$.
(iii) No pair of different elements in $O$ is contained in the same set $S$ for any $S \in \mathcal{S} \backslash \mathcal{R}$. 
(iv) Each %set 
 $S \in \mathcal{S} \backslash \mathcal{R}$ contains one element in $O$ and $r-1$ elements of some set in~$\mathcal{R}$.
\end{lemma}

\begin{proof}
After re-applying Reduction Rule \ref{cleanupSetPacking}, each element in $\mathcal{U}$ is in at least one set $S$ in the reduced collection $\mathcal{S}$. Since $O \subseteq \mathcal{U}$, (i) follows.

To prove claims (ii) and (iii), assume for the sake of contradiction that there is a pair of elements $u_i,u_j \in O$, $i\neq j$, such that $u_i,u_j \in S$. The set $S$ is not in $\mathcal{R}$; otherwise, $u_i,u_j$ would not be in $O$. In addition, if $S$ is disjoint from $\mathcal{R}$ then $S$ could be added to $\mathcal{R}$, contradicting its maximality. Therefore, $S \in \mathcal{S}\backslash \mathcal{R}$ and $S$ overlaps with at least one set $S'$ in $\mathcal{R}$. Since by our assumption, $u_i,u_j$ are both in $S$, $|S \cap S'| \leq r-2$. However, once again $S$ can be added to $\mathcal{R}$ contradicting its maximality. This proves claims (ii) and (iii).

To prove claim (iv), assume that there is a set $S \in \mathcal{S} \backslash \mathcal{R}$ such that $u_i \in S$ and $S$ overlaps in at most $r-2$ elements with each set in $\mathcal{R}$. This once more contradicts the fact that $\mathcal{R}$ is a maximal $(r,r-2)$-set packing. \qed
\end{proof}

We will reduce $O$ by applying similar ideas as in \cite{Moser09}. To this end, we first construct an auxiliary bipartite graph $B=(V_O,V_{\mathcal{S} \backslash \mathcal{R}},E)$ as follows. There is a vertex $v_{o}$ in $V_O$ for each element $o \in O$. For each set $S$ in $\mathcal{S} \backslash \mathcal{R}$ and each subset $P \subsetneq S$ where $|P|=r-1$, if there is at least one element $o \in O$ such that $\{o\} \cup P\in \mathcal{S} \backslash \mathcal{R}$, add a vertex $v_p$ to $V_{\mathcal{S} \backslash \mathcal{R}}$. We say that $v_o$ and $v_p$  \emph{correspond} to $o$ and $P$, respectively. We add an edge $(v_o,v_p)$ to $E$ for each pair $v_o,v_p$ if $\{o\} \cup P\in \mathcal{S}\backslash \mathcal{R}$. Then, we apply the following reduction rule.
%\marginpar{HF: A picture or example would be helpful.}

\begin{redrule}\label{matching}
Compute a maximum matching $M$ in $B$. Let $V'_O$ be the set of unmatched vertices of $V_O$ and $O'$ be the elements of $O \subset \mathcal{U}$ corresponding to those vertices. Likewise, let $\mathcal{S}(O')$ be the sets in $\mathcal{S} \backslash \mathcal{R}$ that contain elements of $O'$. Reduce to $\mathcal{U'}=\mathcal{U} \backslash O'$ and $\mathcal{S'} = \mathcal{S} \backslash \mathcal{S}(O')$.
\end{redrule}

%%Example**************************************************************************************************
\begin{example}\label{bipartiteExample}
For the instance of Example \ref{instanceExample}, $O =\{a,h,k,l\}$, and $$\mathcal{S}\backslash \mathcal{R}=\{\{b,c,e,a\}, \{e,i,f,a\},\{e,g,i,h\},\{i,j,m,k\},\{i,j,m,l\}\}\,.$$ According to our construction, there will be one vertex in $V_{\mathcal{S}\backslash\mathcal{R}}$ for the subsets $\{b,c,e\}$, $\{e,i,f\}$, $\{e,g,i\}$ and $\{i,j,m\}$. Figure \ref{bipartiteGraph} depicts the bipartite graph constructed as described above. The left column represents the vertices in $V_O$ while the right column represents the vertices in $V_{\mathcal{S}\backslash \mathcal{R}}$. A maximum matching in this bipartite graph is highlighted with bold lines. Observe that the element $k$ (or alternatively $l$) will be removed from $\mathcal{U}$ by Reduction Rule \ref{matching}.
\end{example}

\begin{figure}[htb]     
     \centerline{{\includegraphics[scale=0.8]{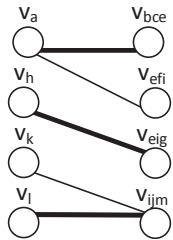}}}
     \caption{The bipartite constructed with $O$ and $\mathcal{S}\backslash \mathcal{R}$ from Example \ref{bipartiteExample}} \label{bipartiteGraph}
\end{figure}

%%End of example ******************************************************************************************************

We show next that Reduction Rule \ref{matching} correctly reduces $\mathcal{U}$ and $\mathcal{S}$. 

%It is improving but is not there yet.

\begin{lemma}\label{unmatched}
$\mathcal{S}$ has a $(k,r,t)$-set packing if and only if $\mathcal{S'}$ has a $(k,r,t)$-set packing.
\end{lemma}

\begin{proof}
Since $\mathcal{S'} \subseteq \mathcal{S}$, if $\mathcal{S'}$ has a $(k,r,t)$-set packing, then $\mathcal{S}$ has one, as well. Suppose that $\mathcal{S}$ has a $(k,r,t)$-set packing $\mathcal{K}$. If $\mathcal{K}$ is not a solution in $\mathcal{S'}$ then at least one set in $\mathcal{K}$ contains an \emph{unmatched element} in $O'$.  

We next show that we can transform $\mathcal{K}$ into a $(k,r,t)$-set packing of $\mathcal{S'}$. Let $\mathcal{K'}=\{S_{o'_1}, \dots, S_{o'_l} \}$ ($l \leq k$) be the sets of $\mathcal{K}$ that contain unmatched elements. By Lemma \ref{contained}, each $S_{o'_i} \in \mathcal{K'}$ contains only one element of $O$. Let us denote as $o'_i$ the element in $O'$ contained in the set $S_{o'_i}$.

%%% -----------------------------------------------------------------------------------------------------------------
%%% Claim 1. There exisits an substitute for the unamtched element.
%%% -----------------------------------------------------------------------------------------------------------------
We next show that for each $o'_i$, there exists an element in $O \backslash O'$ that can replace $o'_i$ to form a new set. By our construction of $B$, there is a vertex $v_{p_{i}} \in V_{\mathcal{S}\backslash{R}}$ representing the subset $P_i = S_{o'_{i}} \backslash \{o'_i\}$, and this vertex $v_{p_i}$ is adjacent to $v_{o'_i} \in V_O$ (the representative vertex of $o'_i$). The vertex $v_{p_i}$ is matched. That is, there exists a matched edge $(v_{o_i},v_{p_i})$ where $v_{o_i} \in V'_{O}$. Otherwise, we could add the edge $(v_{o'_i},v_{p_i})$ to $M$ contradicting the assumption that $M$ is maximum. By our construction, $S_{o_i} = \{o_i\} \cup P_i$ is a set in $\mathcal{S} \backslash \mathcal{R}$. Thus, for each $o'_i$ there exists an element $o_i$ in $O \backslash O'$ that can replace $o'_i$ to form the new set $S_{o_i} = (S_{o'_i} \backslash \{o'_i\}) \cup \{o_i\}$.

%%% -----------------------------------------------------------------------------------------------------------------
%%% Claim 2. The new formly sets do not conflict between each other.
%%% -----------------------------------------------------------------------------------------------------------------
Let us denote as $\mathcal{K^*}$ by these new sets  $\{S_{o_1}, \dots, S_{o_l}\}$. We now claim that these sets pairwise overlap in at most $t$ elements. For the sake of contradiction, suppose that there is a pair $S_{o_i},S_{o_j}$ that conflicts% with each other
, i.e., $|S_{o_i} \cap S_{o_j}| \geq t+1$. Since $S_{o_i} = P_i \cup \{o_i\}$, $S_{o_j} = P_i \cup \{o_j\}$ and $o_i \neq o_j$, the conflict is only possible if $|P_i \cap P_j| \geq t+1$. Each set $S_{o_i}$ is replacing a set $S_{o'_i}$ in $\mathcal{K}$ (a $(k,r,t)$-set packing). Since $S_{o'_i} \cap S_{o_i} = P_i$ and  $S_{o'_j} \cap S_{o_j} = P_j$, $|P_i \cap P_j| \leq t$, as otherwise $\mathcal{K}$ would not be a $(k,t)$ set packing.

%%% -----------------------------------------------------------------------------------------------------------------
%%% Claim 3. We can alter the solutions to ``welcome'' the new sets.
%%% ----------------------------------------------------------------------------------------------------------------

We next show that we can form a solution of $\mathcal{S'}$ by changing $\mathcal{K}$. First, we replace $\mathcal{K'}$ by $\mathcal{K^*}$ in $\mathcal{K}$. If each set in $\mathcal{K^*}$ overlaps in at most $t$ elements with each set in $\mathcal{K} \backslash \mathcal{K'}$, then $\mathcal{K}$ is a $(k,r,t)$-set packing and we are done. Assume that there is a set $S_{o_i} \in \mathcal{K^*}$ that overlaps in more than $t$ elements with some set in $\mathcal{K} \backslash \mathcal{K^*}$ (we already showed that the sets in $\mathcal{K^*}$ pairwise overlap in at most $t$ elements). Since $S_{o_i} = P_i \cup \{o_i\}$, $P_i$ shares at most $t$ elements with any set in $\mathcal{K} \backslash S_{o_i}$. Therefore, if $S_{o_i} = P_i \cup \{o_i\}$ conflicts with some set $S^* \in \mathcal{K} \backslash \mathcal{K^*}$ this is only possible if $o_i \in S^*$. Once again by our construction of $B$, $P^* = S^* \backslash \{o_i\}$ is represented by a vertex $v_{p^*} \in V_{\mathcal{S} \backslash \mathcal{R}}$, and $v_{p^*}$ is adjacent to $v_{o_i}$. We claim that $v_{p^*}$ is a matched vertex, that is, there exists an edge $(v_{o^*},v_{p^*}) \in M$.  Assume by contradiction that $v_{p^*}$ is unmatched. This would imply that we can increase the size of the matching by adding the edges $(v_{o'_i},v_{p_i})$ and $(v_{o_i},v_{p^*})$ and remove the matched edge $(v_{o_i},v_{p_i})$. However, this contradict the assumption that $M$ is maximum. In this way, we can change $S^*$ by removing $o_i$ and adding $o^*$, i.e., $S^*$ becomes $(S^* \backslash \{o_i\}) \cup \{o^*\}$. We can repeat this argument if the new $S^*$ conflicts with another set in $\mathcal{K} \backslash \mathcal{K^*}$. Applying the same argument we can iterate replacing sets in $\mathcal{K}$ until we have a  $(k,r,t)$-set packing in~$\mathcal{S'}$. \qed

\end{proof}

The correctness of our kernelization algorithm (Algorithm \ref{iterative}) is given by Lemmas \ref{solutionsets}-\ref{unmatched}. We next show that this algorithm runs in polynomial time.

%%========================================================================================================
%%RUNNING TIME LEMMAS ========================================================================================
\begin{lemma}\label{iterativetime}
Algorithm \ref{iterative} runs in polynomial time.
\end{lemma}

\begin{proof}
Reduction Rule \ref{cleanupSetPacking} can be computed in $O(n^r)$ time.  An upper bound to compute the set $\mathcal{R}$ is $O(n^r)$. Algorithm \ref{extrasubgraphsred} runs in time $O(n^{r^2})$. In addition, a maximum matching can be computed in polynomial time. \qed
\end{proof}

%%========================================================================================================
%%BOUNDING LEMMAS ========================================================================================
%%========================================================================================================

It remains to determine an upper bound for the number of elements in the reduced universe $\mathcal{U'}$. Given that $\mathcal{U'} = val(\mathcal{S'})$ (Reduction Rule \ref{matching}), we will use $\mathcal{S'}$ to upper-bound $\mathcal{U'}$.

By Line 4 of Algorithm \ref{iterative}, $\mathcal{R} \subseteq \mathcal{S}$. Thus, $\mathcal{S} = (\mathcal{S} \backslash \mathcal{R}) \cup \mathcal{R}$. After Reduction Rule \ref{matching}, some members in $\mathcal{S} \backslash \mathcal{R}$ were removed, and we obtained the reduced collection $\mathcal{S'} = ((\mathcal{S} \backslash \mathcal{R}) \backslash \mathcal{S}(O')) \cup \mathcal{R}$. This is equivalent to $\mathcal{S'} = \mathcal{S} \backslash \mathcal{S}(O')$. Thus, $\mathcal{S'} = (\mathcal{S'} \backslash \mathcal{R}) \cup \mathcal{R}$.

Since $\mathcal{S'} \backslash \mathcal{R} \subseteq \mathcal{S} \backslash \mathcal{R}$, by Lemma \ref{contained} we know that each set in $\mathcal{S'} \backslash \mathcal{R}$ contains one element in $O$ and $r-1$ elements in a set of $\mathcal{R}$. Thus, $val(\mathcal{S'})= O \cup val(\mathcal{R})$.

The elements in $O'$ were removed from $O$ (Reduction Rule \ref{matching}), therefore, $val(\mathcal{S'}) = (O \backslash O') \cup val(\mathcal{R}) = \mathcal{U'}$.

%%***********************************************************************************
%%Upper bound for $\mathcal{R}$ *************************************************************************************
We next provide an upper bound for the size of $\mathcal{R}$. 

\begin{lemma}\label{sizeHgR}
The size of $\mathcal{R}$ is at most $2 r^{r-1}k^{r-t-1}$, i.e., 
%. Thus, 
$|val(\mathcal{R})| \leq  2r^{r} k^{r-t-1}$. % elements.
\end{lemma}

\begin{proof}
Let $\mathcal{M}$ be a maximal $(r,t)$-set packing of $\mathcal{R}$. Each set $S \in \mathcal{R}$ is either in $\mathcal{M}$ or it overlaps with at least one set $S' \in \mathcal{M}$ in at least $t+1$ elements (i.e., $|S \cap S'| \geq t+1$). Otherwise, $\mathcal{M}$ would not be a maximal $(r,t)$-set packing. 

For each set $S \in \mathcal{R}$, each subset of $S$ with $t+1$ elements is contained in at most $f(t+1)$ sets in $\mathcal{R}$ (a direct consequence of Lemma \ref{solutionsets}). Therefore, we can use the subsets of size $t+1$ contained in the members of $\mathcal{M}$ to upper-bound the number of sets in $\mathcal{R}$.

For each set in $\mathcal{M}$, there are $\binom{r}{t+1}$ subsets of size $t+1$. Thus, $|\mathcal{R}|$ $\leq$ $|\mathcal{M}|$ $\binom{r}{t+1}$ $f(t+1)$. 

If $|\mathcal{M}| \geq k$ then $\mathcal{M}$ is a $(k,r,t)$-set packing of $\mathcal{S}$ and the algorithm outputs a subset of $\mathcal{M}$ and stops in Line 14. Thus, $|\mathcal{M}| \leq k-1$. On the other hand,  $f(t+1)= \sum_{j=0}^{t_{Ini}-(t+1)+1} [(r-t)(k-1)]^j$ $<$ $2((r-t)(k-1))^{t_{Ini}-(t+1)+1}$. 

In this way, $|\mathcal{R}|$ $\leq$ $r^{t+1}$ $k-1$ $2((r-t)(k-1))^{t_{Ini}-(t+1)+1}$. Given that $t_{Ini}=r-2$, $|\mathcal{R}| \leq 2r^{r-1} k^{r-t-1}$ and $val(\mathcal{R}) \leq 2r^r k^{r-t-1}$. \qed
\end{proof}

%%************************************************************************************
%%Upper bound for $O$ ***************************************************************************************

Now, we  bound the size of $O \backslash O'$. 

\begin{lemma}
After applying Reduction Rule \ref{matching}, there are at most $2 r^r k^{r-t-1}$ elements in $O \backslash O'$.
\end{lemma}

%\marginpar{HF What is a partite?}

\begin{proof}
After applying Reduction Rule \ref{matching}, $|V'_O| \leq |V_{\mathcal{S}\backslash \mathcal{R}}|$. By our construction of the set  $V_{\mathcal{S}\backslash \mathcal{R}}$, there are at most $\binom{r}{r-1}|
\mathcal{R}|=$$r|\mathcal{R}|$ vertices in $V_{\mathcal{S}\backslash \mathcal{R}}$. By Lemma \ref{sizeHgR}, $|\mathcal{R}| \leq 2 r^{r-1} k^{r-t-1}$ and the lemma follows. \qed
\end{proof}

%==================================================================================
%Improved kernel
%==================================+===============================================
Lemma \ref{iterativetime}, the $2 r^r k^{r-t-1}$ elements in $O \backslash O'$ together with the $2r^{r} k^{r-t-1}$ elements in $\mathcal{R}$ give the following result.

\begin{theorem}
The $r$-Set Packing with $t$-Overlap possesses a problem kernel with  $O(r^r k^{r-t-1})$ elements from the given universe.
\end{theorem}

%% file: Pairwise_Vertex_ToPairwiseSP.tex
%************************************************************************************************
%Problem Definition ------------------------------------------------------------------------------

%\medskip
%\fbox{
%\parbox{11cm}{
%\textbf{The $\mathcal{H}$-Packing with $t$-Overlap problem}
%    
%    \noindent \emph{Input}: A graph $G$, and a positive integer $k$.
%
%    \noindent \emph{Parameter}: $k$
%
%    \noindent \emph{Question}: Does $G$ contain a $(k,t)$-$\mathcal{H}$-Packing, i.e., a set of at least $k$ subgraphs $\mathcal{K}=\{H_1, \dots ,H_k\}$ where each $H_i$ is isomorphic to some graph from  $\mathcal{H}$ and $|V(H_i) \cap V(H_j)| \leq t$ for any pair $H_i,H_j$ with $i\neq j$?%where $t \geq 0$?
%}}
%\medskip

We assume that each $H \in \mathcal{H}$ is an arbitrary graph with at least $t+1$ vertices. Otherwise, we just add the set of $\mathcal{H}$-subgraphs of $G$ each with at most $t$ vertices straight to a $(k,r,t)$-$\mathcal{H}$-packing and decrease the parameter $k$ by the size of that set. 

%\marginpar{HF: Here and in the following, do we want to mention these reductions at all? This is automatically done after transformning the instance to Set Packing. \textcolor{blue}{Let me know if this simplication is better, or if we should cut more}. Well, at least there would be no harm in omitting them, as the Set Packing procedure would do this anyways.} 

The kernelization algorithm for the $\mathcal{H}$-Packing with $t$-Overlap problem corresponds to Algorithm \ref{GenericKernelization} with input: $G$, $\mathcal{H}$, $k$, ``$r$-Set Packing with $t$-Overlap'', Transformation \ref{VertexPackingToOSP}, Algorithm \ref{iterative}, and Transformation \ref{FromUniverseToGraph}. Since Algorithm \ref{iterative} returns a reduced universe with $O(r^r k^{r-t-1})$ elements, the order of $G'$ is $O(r^r k^{r-t-1})$. This allow us to conclude.

%We reduce this problem to a kernel by following similar steps as in Subsection \ref{membershipGraphs}.  First, we create an instance of the $r$-Set Packing with $t$-Overlap problem with Transformation \ref{VertexPackingToOSP}. Next, we run Algorithm \ref{iterative} on this constructed instance. That algorithm returns a reduced universe $\mathcal{U'}$ and a reduced collection $\mathcal{S'}$ which we will use to obtain the reduced graph $G' = G[\mathcal{U'}]$. This allow us to conclude.

\begin{theorem}\label{pairwisevertexkernel}
The $\mathcal{H}$-Packing with $t$-Overlap possesses a problem kernel with $O(r^r k^{r-t-1})$ vertices, where $r=r(\mathcal{H})$.
\end{theorem}

%% file: Pairwise_Induced_ToPairwiseSP.tex
The induced version seeks for at least $k$ induced $\mathcal{H}$-subgraphs in $G$ that pairwise overlap in at most $t$ vertices. To achieve a problem kernel, we redefine $\mathcal{S}$ of Transformation \ref{VertexPackingToOSP} by adding a set $S=V(H)$ per each each \textbf{induced} $\mathcal{H}$-subgraph $H$ of $G$.

\begin{theorem}
The Induced-$\mathcal{H}$-Packing with $t$-Overlap has a problem kernel with $O(r^r k^{r-t-1})$ vertices, where $r=r(\mathcal{H})$.
\end{theorem}

%% file: Pairwise_Edge_ToPairwiseSP.tex
We introduce the $\mathcal{H}$-Packing with $t$-Edge Overlap Problem to regulate the overlap between the edges of the $\mathcal{H}$-subgraphs instead of the vertices. Recall that there is an upper bound $m(\mathcal{H})$ on the number of edges of each graph occurring in $\mathcal{H}$. Also as in Subsection \ref{EdgeMembershipSection},
$\mathcal{H}$ hosts only graphs without isolated vertices.%\marginpar{HF: ok?}
%Again, let $t\geq 0$ in the following definition.

%*********************************************************************************
%Problem Definition --------------------------------------------------------------

%\medskip
%\fbox{
%\parbox{11cm}{
%\textbf{The $\mathcal{H}$-Packing with $t$-Edge-Overlap problem}
%    
%    \noindent \emph{Input}: A graph $G$, and a positive integer $k$.
%
%    \noindent \emph{Parameter}: $k$
%
%    \noindent \emph{Question}: Does $G$ contain a $(k,t)$-Edge-$\mathcal{H}$-Packing, i.e., a set of at least $k$ subgraphs $\mathcal{K}=\{H_1, \dots ,H_k\}$ where each $H_i$ is isomorphic to a graph $H \in \mathcal{H}$, and $|E(H_i) \cap E(H_j)| \leq t$ for any pair $H_i,H_j$ with $i\neq j$?% where $t \geq 0$?
%}}
%\medskip

%Similar as in Subsection \ref{membershipGraphs}, $G$ and each $H \in \mathcal{H}$ do not contain isolated vertices.
We assume that each $H \in \mathcal{H}$ has at least $t+1$ edges. Otherwise, we just add the set of $\mathcal{H}$-subgraphs of $G$ each with at most $t$ edges straight to a $(k,r,t)$-edge-$\mathcal{H}$-packing and decrease the parameter $k$ by the size of that set.

By following a similar scheme as in the vertex-version, the kernelization algorithm for the $\mathcal{H}$-Packing with $t$-Edge Overlap Problem consists of Algorithm \ref{GenericKernelization} with input: $G$, $\mathcal{H}$, $k$, ``$r$-Set packing with $t$-Overlap'', Transformation \ref{edgePackingToOSP}, Algorithm \ref{iterative}, and Transformation \ref{FromUniverseToGraph_v3}. Since Algorithm \ref{iterative} returns a universe with $O(r^r k^{r-t-1})$ elements and $G'$ does not contain isolated vertices (Lemma \ref{GNoIsolated}), $|V(G')| \leq 2|\mathcal{U'}| = O(r^r k^{r-t-1})$. 

%Once more we apply similar ideas as in Subsection \ref{membershipGraphs}. We create an instance of the $r$-Set Packing problem with $t$-Overlap using Transformation \ref{edgePackingToOSP}. Next, we run Algorithm \ref{iterative}  on the constructed instance of the $r$-Set Packing with $t$-Overlap. That algorithm returns a reduced universe $\mathcal{U'}$ and a reduced collection $\mathcal{S'}$. Let $V(\mathcal{U'})$ denote the set  of all end-points of edges in $\mathcal{U'}$. Then, the reduced graph is $G' = G[V(\mathcal{U'})]$. 

%Given that $G$ do not contain isolated vertices (Lemma \ref{GNoIsolated}), we have the next result.
 
\begin{theorem}\label{edgepackingkernel} 
The $\mathcal{H}$-Packing with $t$-Edge-Overlap has a problem kernel with $O(r^{r} k^{r-t-1})$ vertices, where $r=m(\mathcal{H})$.
\end{theorem}

%% file: Pairwise_Cliques_ToPairwiseSP.tex
We next obtain an $O(r^rk^{r-t-1})$ kernel, where $r=r(\mathcal{H})$, for the $\mathcal{H}$-Packing with $t$-Edge-Overlap problem when $\mathcal{H}$ is a family of cliques. Each clique in $\mathcal{H}$ has at least $t+1$ vertices and at most $r(\mathcal{H})$ vertices, i.e., a $\mathcal{H}=\{K_{t+1}, \dots , K_{r(\mathcal{H})} \}$. 

We next present the formal definition of the problem. To avoid ambiguity, we will precede the word overlap either by ``edge'' or ``vertex'' to differentiate both versions of the problem. Let $t\geq0$ in the following definition.

\vspace{0.25cm}
\begin{center}
\fbox{
\parbox{13.5cm}{
\textbf{The $K_r$-Packing with $t$-Edge-Overlap problem}
    
    \noindent \emph{Input}: A graph $G$, and a non-negative integer $k$.

    \noindent \emph{Parameter}: $k$

    \noindent \emph{Question}: Does $G$ contain a \emph{$(k,r,t)$-edge-$K_r$-packing}, i.e., a set $\mathcal{K}=\{H_1, \dots ,H_k\}$ where each $H_i$ is a clique with at least $t+1$ and at most $r$ vertices and $|E(H_i) \cap E(H_j)| \leq t$ for any pair $H_i,H_j$ with $i\neq j$?%where $t \geq 0$?
}}
\end{center}
\medskip
		
Note that when $t=0$, the cliques in $\mathcal{K}$ are pairwise edge-disjoint. This implies that they overlap in at most $1$ vertex. In the same way if $t=1$, any pair of cliques in $\mathcal{K}$ overlaps in at most 2 vertices. %\marginpar{HF: I moved this up, as it does not belong to the following proof IMHO, ok. Thanks.}

We next show that also in general, the $K_r$-Packing with $t$-Edge-Overlap Problem relates to the vertex-overlap version of the problem. More precisely, we will show that any $(k,r,t)$-edge-$K_r$-packing in $G$ is a $(k,r,t')$-$K_r$-packing in $G$  where $t'$ is the largest integer such that

\begin{equation}\label{valueoft}
    \frac{t'(t'-1)}{2} \leq t
\end{equation}

\noindent
for $t \geq 2$. As mentioned above, if $t=0$ then $t'=1$ and if $t=1$ then $t'=2$.

\begin{lemma}
Let $\mathcal{K}$ be any $(k,r,t)$-edge $K_r$-packing in $G$. Any pair of $K_r$'s in $\mathcal{K}$ overlaps in at most $t'$ vertices.
\end{lemma}

\begin{proof}
Any pair of cliques in $\mathcal{K}$ shares a $K_s$ for some $0 \leq s \leq r-1$. In other words, the set of vertices $S$ shared between any pair of cliques in $\mathcal{K}$ induces a $K_s$ in $G$. We next prove that $s \leq t'$ where $t'$ is defined as in Equation \ref{valueoft}.

Since any pair of cliques in $\mathcal{K}$ overlap in at most $t$ edges (otherwise, $\mathcal{K}$ would not be a $(k,t)$-edge-$K_r$-packing), then $|E(G[S])| \leq t$. The number of edges in $G[S]$ is at most $\frac{|V(S)|(|V(S)-1|)}{2} \leq t$. Therefore, $|V(S)|$ is at most the largest integer $t'$ satisfying Equation \ref{valueoft}.
%such that 
%$\frac{|V(S)|(|V(S)-1|)}{2} \leq t$.
\qed
\end{proof}

\begin{coro}\label{equivalent}
Any $(k,r,t)$-edge-$K_r$-packing of $G$ is a $(k,r,t')$-$K_r$-packing of $G$ and vice versa where $t'$ is defined as Equation \ref{valueoft}.
\end{coro}

\noindent
By Corollary \ref{equivalent}, we can have the next result. 

\begin{theorem}
The $K_r$-Packing with $t$-Edge-Overlap problem has a kernel with $O(r^r k^{r-t'-1})$ vertices, where $t'$ is defined as Equation \ref{valueoft} and $r=r(\mathcal{H})$.
\end{theorem}

%% file: Conclusions.tex
We commenced the study on the parameterized algorithmics of packing problems allowing overlap, focusing on kernelization issues. This leads to whole new families of problems, whose problem names contain the constant $t$ and the finite set of objects (mostly graphs) $\mathcal{H}$ from which the constant $r(\mathcal{H})$ can be derived. In addition, we have a natural parameter $k$.

We list here some of the open problems in this area. Our kernel bounds are not known to be tight. In particular, we have shown no lower bound results. 
We refer (in particular) to \cite{DelMar2012,HerWu2012} for such results for disjoint packings in graphs and hypergraphs.
Recently, Giannopoulou et al. introduced the notion of \emph{uniform kernelization}~\cite{GiaJLS2014} for problem families similar to what we considered. This basically raises the question if polynomial kernel sizes could be proven such that the exponent of the kernel bound does not depend (in our case) on $r$ or on $t$. Notice that for $t$-Membership, we somehow came half the way, as we have shown kernel sizes that are uniform with respect to $t$. In particular, $\{K_3\}$-Packing with $t$-Membership has a uniform kernelization. This could be interesting on its own, as only few examples of uniform kernelizations are known. Very recent work by Marx and his colleagues has shown renewed interest in dichotomy results like the one that we presented in this paper. Apart from the natural questions mentioned above, the case of admitting an infinite number of graphs in the family $\mathcal{H}$ is another natural path for future research.
 In both respects, Jansen and Marx~\cite{JanMar2014} have obtained very interesting results on classical graph packing problems.
 Finally, it might be interesting to derive better parameterized algorithms and kernels for concrete packing problems allowing overlaps.

%% file: paper.bbl
\begin{thebibliography}{10}
\providecommand{\url}[1]{\texttt{#1}}
\providecommand{\urlprefix}{URL }

\bibitem{Faisal10}
Abu-Khzam, F.N.: An improved kernelization algorithm for $r$-{S}et {P}acking.
  Information Processing Letters  110(16),  621--624 (2010)

\bibitem{BanKhu2001}
Banerjee, S., Khuller, S.: A clustering scheme for hierarchical control in
  multi-hop wireless networks. In: Proceedings of 20th Joint Conference of the
  IEEE Computer and Communications Societies (INFOCOM 2001). vol.~2, pp.
  1028--1037. IEEE Society Press (2001)

\bibitem{Bodlaender08}
Bodlaender, H.L., Thomass\'e, S., Yeo, A.: Analysis of {D}ata {R}eduction:
  Transformations give evidence for non-existence of polynomial kernels. Tech.
  Rep. UU-CS-2008-030, Department of Information and Computer Sciences, Utrecht
  University (2008)

\bibitem{CaprRiz2002}
Caprara, A., Rizzi, R.: Packing triangles in bounded degree graphs. Information
  Processing Letters  84(4),  175 -- 180 (2002)

\bibitem{Chen12}
Chen, J., Fernau, H., Shaw, P., Wang, J., Yang, Z.: Kernels for {P}acking and
  {C}overing {P}roblems. In: Snoeyink, J., Lu, P., Su, K., Wang, L. (eds.)
  Frontiers in Algorithmics and Algorithmic Aspects in Information and
  Management. LNCS, vol. 7825, pp. 199--211. Springer, Heidelberg (2012)

\bibitem{DelMar2012}
Dell, H., Marx, D.: Kernelization of packing problems. In: Rabani, Y. (ed.)
  Proceedings of the Twenty-Third Annual ACM-SIAM Symposium on Discrete
  Algorithms, SODA. pp. 68--81. SIAM (2012)

\bibitem{DorTar97}
Dor, D., Tarsi, M.: Graph decomposition is {NP}-complete: A complete proof of
  {H}olyer's conjecture. {SIAM} Journal on Computing  26(4),  1166--1187 (1997)

\bibitem{Fellows09}
Fellows, M., Guo, J., Komusiewicz, C., Niedermeier, R., Uhlmann, J.:
  Graph-based data clustering with overlaps. Discrete Optimization  8(1),
  2--17 (2011)

\bibitem{Fellows04b}
Fellows, M., Knauer, C., Nishimura, N., Ragde, P., Rosamond, F., Stege, U.,
  Thilikos, D., Whitesides, S.: {F}aster {F}ixed-{P}arameter {T}ractable
  {A}lgorithms for {M}atching and {P}acking {P}roblems. Algorithmica  52(2),
  167--176 (2008)

\bibitem{Fellows04a}
Fellows, M., Heggernes, P., Rosamond, F., Sloper, C., Telle, J.A.: Finding $k$
  {D}isjoint {T}riangles in an {A}rbitrary {G}raph. In: Hromkovi\v{c}, J.,
  Nagl, M., Westfechtel, B. (eds.) The 30th Workshop on Graph-Theoretic
  Concepts in Computer Science. LNCS, vol. 3353, pp. 235--244. Springer,
  Heidelberg (2004)

\bibitem{Fernau09}
Fernau, H., Raible, D.: A {P}arameterized {P}erspective on {P}acking {P}aths of
  {L}ength {T}wo. Journal of Combinatorial Optimization  18(4),  319--341
  (2009)

\bibitem{GiaJLS2014}
Giannopoulou, A.C., Jansen, B.M.P., Lokshtanov, D., Saurabh, S.: Uniform
  kernelization complexity of hitting forbidding minors (2014), unpublished,
  see http://www.win.tue.nl/\~{}bjansen/publications.html

\bibitem{HerWu2012}
Hermelin, D., Wu, X.: Weak compositions and their applications to polynomial
  lower bounds for kernelization. In: Rabani, Y. (ed.) Proceedings of the
  Twenty-Third Annual ACM-SIAM Symposium on Discrete Algorithms, SODA. pp.
  104--113. SIAM (2012)

\bibitem{Hol81}
Holyer, I.: The {NP}-completeness of some edge-partition problems. {SIAM}
  Journal on Computing  10(4),  713--717 (1981)

\bibitem{JanMar2014}
Jansen, B.M.P., Marx, D.: Characterizing the easy-to-find subgraphs from the
  viewpoint of polynomial-time algorithms, kernels, and {T}uring kernels. CoRR
  abs/1410.0855 (2014)

\bibitem{Kirk78}
Kirkpatrick, D., Hell, P.: On the completeness of a generalized matching
  problem. In: Proceedings of the Tenth Annual {ACM} Symposium on Theory of
  Computing ({STOC}). pp. 240--245 (1978)

\bibitem{MicaliVazirani}
Micali, S., Vazirani, V.V.: An {{$O(\sqrt{|V|}|E|)$}} {A}lgorithm for {F}inding
  {M}aximum {M}atching in {G}eneral {G}raphs. In: Proceedings of the 21st
  Annual Symposium on Foundations of Computer Science. pp. 17--27. SFCS '80,
  IEEE Computer Society (1980)

\bibitem{Moser09}
Moser, H.: A {P}roblem {K}ernelization for {G}raph {P}acking. In: Nielsen, M.,
  Ku\v{c}era, A., Miltersen, P.B., Palamidessi, C., T\r{u}ma, P., Valencia, F.
  (eds.) SOFSEM 2009. LNCS, vol. 5404, pp. 401--412. Springer, Heidelberg
  (2009)

\bibitem{Palla05}
Palla, G., Der\'{e}nyi, I., Farkas, I., Vicsek, T.: Uncovering the overlapping
  community structure of complex networks in nature and society. Nature
  435(7043),  814--818 (2005)

\bibitem{Prieto06}
Prieto, E., Sloper, C.: Looking at the stars. Theoretical Computer Science
  351(3),  437--445 (2006)

\bibitem{Packing14}
Romero, J., L\'opez-Ortiz, A.: The $\mathcal{G}$-{P}acking with $t$-{O}verlap
  {P}roblem. In: Pal, S.P., Sadakane, K. (eds.) WALCOM 2014. LNCS, vol. 8344,
  pp. 114--124. Springer, Heidelberg (2014)

\bibitem{BstAlgorithm14}
Romero, J., L\'opez-Ortiz, A.: A {P}arameterized {A}lgorithm for {P}acking
  {O}verlapping {S}ubgraphs. In: Hirsch, E.A., Kuznetsov, S.O., Pin, J.E.,
  Vereshchagin, N.K. (eds.) CSR 2014. LNCS, vol. 8476, pp. 325--336. Springer,
  Heidelberg (2014)

\bibitem{Shiloach81}
Shiloach, Y.: Another look at the degree constrained subgraph problem.
  Information Processing Letters  12(2),  89--92 (1981)

\bibitem{Wang10}
Wang, J., Ning, D., Feng, Q., Chen, J.: An improved kernelization for
  {$P_2$}-packing. Information Processing Letters  110(5),  188--192 (2010)

\end{thebibliography}
